\documentclass{article}
\usepackage{amsmath,bm}
\usepackage{float}
\usepackage{cite}
\usepackage{xcolor}
\usepackage{xstring}
\usepackage{caption}
\usepackage[english]{babel}
\usepackage[utf8x]{inputenc}
\usepackage{amsfonts}
\usepackage{amssymb}
\usepackage{amsthm}
\usepackage{theoremref}
\usepackage{enumitem}
\usepackage{subfig}
\usepackage{tikz}
\usepackage[export]{adjustbox}
\usepackage{hyperref}
\hypersetup{
  colorlinks   = true, 
  urlcolor     = blue, 
  linkcolor    = blue, 
  citecolor   =  green 
}
\usepackage{placeins} 
\usepackage{ifthen}
\usepackage{mathtools,bm}
\usepackage[linesnumbered,algoruled,boxed,lined]{algorithm2e}
\usepackage{pdfsync}
\def\BibTeX{{\rm B\kern-.05em{\sc i\kern-.025em b}\kern-.08em
    T\kern-.1667em\lower.7ex\hbox{E}\kern-.125emX}}

\newtheorem{theorem}{Theorem}
\newtheorem{lemma}{Lemma}
\newtheorem{proposition}{Proposition}
\newtheorem{corollary}{Corollary}

\newtheorem{remark}{Remark}
\newtheorem{problem}{Problem}
\newtheorem{example}{Example}
\newtheorem{definition}{Definition}
\newtheorem{conj}{Conjecture}

\newcommand{\F}{\mathbb{F}}
\newcommand{\Fq}{\mathbb{F}_{q}}
\newcommand{\Fqm}{\mathbb{F}_{q^m}}

\newcommand{\ba}{\mathbf{a}}
\newcommand{\bb}{\mathbf{b}}

\newcommand{\bs}{\mathbf{s}}
\newcommand{\be}{\mathbf{e}}
\newcommand{\bee}{\bm\varepsilon}

\newcommand{\bg}{\mathbf{g}}

\newcommand{\bx}{\mathbf{x}}
\newcommand{\by}{\mathbf{y}}
\newcommand{\bz}{\mathbf{z}}

\newcommand{\bu}{\mathbf{u}}
\newcommand{\bv}{\mathbf{v}}
\newcommand{\bw}{\mathbf{w}}
\newcommand{\oo}{\bullet}
\newcommand{\og}{\textcolor{gray}{\bullet}}
\newcommand{\bzero}{\mathbf{0}}

\newcommand{\calA}{\mathcal{A}}
\newcommand{\calB}{\mathcal{B}}
\newcommand{\calC}{\mathcal{C}}

\newcommand{\calE}{\mathcal{E}}
\newcommand{\calF}{\mathcal{F}}

\newcommand{\calH}{\mathcal{H}}

\newcommand{\calS}{\mathcal{S}}

\newcommand{\calU}{\mathcal{U}}

\newcommand{\calV}{\mathcal{V}}

\newcommand{\cn}{\color{black}}

\newcommand{\gbinom}[2]{\genfrac{[}{]}{0pt}{}{#1}{#2}}

\DeclareMathOperator{\Rank}{Rank}
\DeclareMathOperator{\Sol}{Sol}

\DeclareMathOperator{\RowSpan}{RowSpan}
\counterwithout{table}{section}

\newcommand{\Span}[2]{\langle #2 \rangle_{#1}}

\newcommand{\Wildcard}{*}
\newcommand{\Slice}[3]{
    \IfEqCase{#2}{
        {1}{#1_{#3, \Wildcard, \Wildcard}}
        {2}{#1_{\Wildcard, #3, \Wildcard}}
        {3}{#1_{\Wildcard, \Wildcard, #3}}
    }
}

\newif\ifdraft
\draftfalse 

\title{Bounded-Degree \\
	Low-Rank Parity-Check Codes 
	}

\author{Ermes Franch, Chunlei Li \\ \texttt{Email: chunlei.li@uib.no}}

\begin{document}
\maketitle

\begin{abstract}
Low-rank parity-check (LRPC) codes are the rank-metric analogue of low-density parity-check codes and they found important applications in code-based cryptography.
In this paper we investigate a sub-family of LRPC codes, which 
have a parity-check matrix defined over a subspace $\calV_{\alpha,d}=\Span{\Fq}{1,\alpha, \ldots, \alpha^{d-1}}\subsetneq \Fqm$, where $\Fqm$ is the finite field of $q^m$ elements and $d$ is a positive integer significantly smaller than $m
$; and they are termed bounded-degree LRPC (BD-LRPC) codes. 
These codes are the same as the standard LRPC codes of density $2$ when the degree $d=2$, while for degree $d>2$ they constitute a proper subset of LRPC codes of density $d$.  
Exploiting the structure of $\calV_{\alpha,d}$, the BD-LRPC codes of degree $d$
can uniquely correct errors of rank weight $r$ when $n-k \geq r + u$ for certain $u \geq 1$, in contrast to the condition $n-k\geq dr$ required for the standard LRPC codes. This underscores 
the  superior decoding capability of the BD-LRPC codes. 
Moreover,
as the code length $n\rightarrow \infty$, when $n/m\rightarrow 0$, the BD-LRPC codes with a code rate of $R=k/n$ can be uniquely decodable with radius $\rho=r/n$ approaching the Singleton bound $1-R$ by letting $\epsilon=u/n\rightarrow 0$; and when $n/m$ is a constant, the BD-LRPC codes can have unique decoding radius $\rho = 1-R-\epsilon $ for a small $\epsilon$, allowing for $\rho>(1-R)/2$ with properly chosen parameters.
 This superior decoding capability is theoretically proved for the case $d=2$ and confirmed by experimental results for $d>2$.
\end{abstract}

\section{Introduction}
Rank-metric codes, a type of error-correction codes equipped with rank metric, have undergone significant advancements over the past decades and have a wide range of applications, including network coding \cite{SilvaKschischangKoetter}, criss-cross error correction \cite{roth1991maximum}, and cryptography \cite{GPT}. Many important results for Hamming metric codes 
have been analogously extended to their counterparts in rank metric \cite{RankBook1, RankBook2}, such as the Singleton bound, the MacWilliams identities and efficient unique decoding algorithms from Reed-Solomon codes to Gabidulin codes. In 2013 Gaborit, Murat, Ruatta, and Z\'{e}mor proposed low-rank parity-check codes \cite{gaborit2013}, which can be deemed as an equivalent of the
low-density parity-check (LDPC) codes \cite{gallager1962} for Hamming metric. 
LRPC codes allow for unique decoding through a probabilistic, polynomial-time algorithm and exhibit sounding randomness that is highly valued in cryptographic applications.
Since then they have been adopted in several post-quantum cryptographic primitives, including RankSign \cite{RankSign2014}, identity-based encryption \cite{gaborit2017}, ROLLO\cite{ROLLO}, the signature scheme Durandal \cite{Aragon2019}, etc. 

Decoding capability is a crucial aspect of rank metric codes for their applications. 
The majority of efficient unique decoding algorithms for 
Reed-Solomon codes have been extended to the linearized setting for Gabidulin codes,  allowing for uniquely decoding errors 
of ranks $r$ within half of the minimum rank distance
\cite{Gabidulin1985, RankBook1, RankBook2}. However, when the rank of errors is beyond $\frac{n-k}{2}$,
it is
challenging to efficiently decode 
it for rank metric codes \cite{WachterZeh,Xing-Yuan2018}. \cn
Recently the decoding capability of LRPC codes was further investigated in \cite{LRPC-2019-TIT}, where for density $d=2$ and $m\geq 3rd-2$, LRPC codes are shown to be uniquely decodable with decoding radius $\rho=\frac{2}{3}(1-R)$. In \cite{FranchLi} we proposed improved decodings for LRPC codes when $n-k<rd$.
Motivated by the improved decoding capability of LRPC codes, in this paper we proposed a subfamily of LRPC codes, which allow for a parity-check matrix $H$ whose entries are randomly drawn from a subspace $\calV_{\alpha, d}=\Span{\Fq}{1,\alpha, \dots, \alpha^{d-1}}\subsetneq \Fqm$ for a positive integer $d$ significantly smaller than $m$. 
We term these codes bounded-degree LRPC (BD-LRPC) codes according to the structure of $\calV_{\alpha, d}$.
This particular structure allows for efficient decoding of errors even when $\calS\subsetneq \calV_{\alpha, d}.\calE$, where
$\calE, \calS$ are the $\Fq$-subspaces generated by the components of the error $\be \in \Fqm^n$ and the syndrome $\bs = \be H^\intercal$, respectively,
and 
$\calV_{\alpha, d}.\calE$ is the product of two subspaces.
We observe that, when multiplied with a subspace $\calV_{\alpha, t}$, the subspaces
$\calS$ and $\calV_{\alpha, d}.\calE$ will be expanded in different paces, leading to 
$\calV_{\alpha, t}.\calS = \calV_{\alpha, t}.\calV_{\alpha, d}.\calE = \calV_{\alpha, d+t-1}.\calE$ for certain positive integer $t$. 
This equality is shown to hold with a high probability with technical calculations in Sections \ref{sec:Successful} and \ref{sec:ExpRowSub}. 
Similarly to the decoding for standard LRPC codes, from the equality $\calV_{\alpha, t}.\calS = \calV_{\alpha, d+t-1}.\calE$ we can recover the error support $\calE$ and then uniquely correct the error $\be$ with the syndrome equations given by $\bs=\be H^\intercal$ when 
$d(n-k) \geq nr$. 
With the aforementioned expansion approach, the asymptotic decoding capability of the proposed BD-LRPC codes is good: 
for a code rate $R$, they can be uniquely decoded with a high probability of decoding radius $\rho= 1-R-\epsilon$, where $\epsilon=\frac{u}{n}$;
as $n$ approaches infinity, when the ratio $b=\frac{n}{m}\rightarrow 0$, $\epsilon$ can be arbitrarily small when $u$ is chosen as a fixed constant;
and when $b=\frac{n}{m}<\frac{2}{11\rho}$ is a constant, it allows the parameter $\epsilon<(1-R)/2$,
which implies that the unique decoding radius $\rho> (1-R)/2$. 
The aforementioned decoding capability of BD-LRPC codes is theoretically shown for $d=2$ or $d>2$ and $r=1$. For the general case that $d>2$ and $r>1$, 
the capability also appear to hold based on our observation from experimental results.
Unfortunately, our techniques are not sufficient to prove the observation, and we propose Conjecture \ref{con_d>2} in Section \ref{sec:bdlrpc}, cordially inviting interested readers to attack this problem. 
It is worth noting that the BD-LRPC codes introduced in this paper serve as another solution to the open problem raised by Xing and Yuan in \cite{Xing-Yuan2018}. 

The remainder of this paper is organized as follows. Section \ref{sec:prel} recalls the basics and auxiliary results on rank-metric codes and LRPC codes. 
Section \ref{sec:bdlrpc} introduces the BD-LRPC codes and 
shows that the decoding capability of these codes approaches the Singleton bound. The main result is given in Theorem \ref{Th0}, for which important theoretical arguments and technical calculations are given in Proposition \ref{prop:0} in Section \ref{sec:Successful} and in Theorems \ref{th:|C_k^ur|} and \ref{th:|Gammak^t|} in Section \ref{sec:ExpRowSub}, respectively. Section \ref{sec:conclusion} concludes our works in this paper.


\section{Preliminaries} \label{sec:prel}
Let $q$ be a prime power. It is well known that the finite field $\Fqm$ with $q^m$ elements
can be seen as an $m$-dimensional vector space over the base field $\Fq$. 
We will denote $\Fq$-linear subspaces of $\Fqm$ by calligraphic letters.
A vector of $\Fqm^n$ will be denoted by a lower-case letter in bold, and 
its components will be denoted by the same letter in normal form, 
e.g., a vector $\bv =(v_1,\dots, v_n)\in \Fqm^n$.
For two positive integers $l\leq n$, the notation $[l \ldots n]$ denotes the set $\{l, \ldots, n\} \subseteq \mathbb{Z}$,
and we will use the shorthand notation $[n]$ for $l=1$.

Given a set $S\subseteq\Fqm$, the \textit{rank support} of 
$S$, denoted by $\Span{\Fq}{S}$, is defined as the $\Fq$-linear span of all the elements in $S$.
This notion can be naturally applied to vectors and matrices over $\Fqm$, namely, the support of a vector $\bv\in\Fqm^n$ is given by 
$\Span{\Fq}{\bv} = \Span{\Fq}{v_1,\dots, v_n}$ and the support of a matrix $A\in\Fqm^{r\times t}$ is given by 
$\Span{\Fq}{A} = \Span{\Fq}{a_{i,j} \mid i\in[r], j\in [t]}$.
For an element $\alpha \in \Fqm$ we denote $\alpha S = \{\alpha s \mid s \in S\}$.
Given $\calV = \Span{\Fq}{v_1, \ldots, v_n} \subseteq \Fqm$  an $\Fq$-linear subspace of $\Fqm$ of dimension $n$, it is clear that $\alpha \calV = \Span{\Fq}{\alpha v_1, \ldots, \alpha v_n}$  for any $\alpha \neq 0$ is still an $\Fq$-linear subspace of dimension $n$. 
We will denote by $\Fq^{n\times k}$ the set of all $n\times k$ matrices over $\Fq$. 

Let $\bv = (v_1,\ldots, v_d)$ and $\bu = (u_1,\ldots, u_r)$ be two vectors over $\Fqm$ consisting of $\Fq$-linearly independent components. 
We define the \textit{product of the two subspaces} $\calV = \Span{\Fq}{\bv}$ and $\calU=\Span{\Fq}{\bu}$ as 
$$
\calV.\calU = \Span{\Fq}{\bv \otimes \bu}=\Span{\Fq}{(v_1 u_1, \ldots, v_1u_r, \ldots, v_du_1,\ldots, v_du_r)}.
$$ 
It is readily seen that the product  $\calV.\calU \subseteq\Fqm$ has $\Fq$-dimension at most $dr$.

Denote $A_q(n,k) = \prod_{i = 0}^{k-1} (q^n - q^i)$. It is well known  \cite{goldman1970} that  the number of distinct subspaces of dimension $k$ in $\F_q^n$ is given by
the Gaussian binomial coefficient $\gbinom{n}{k}_q=A_q(n,k)/A_q(k,k)$. 
In addition, it can be shown that the number of matrices in $\Fq^{n\times m}$ with rank $k\leq \min\{n, m\}$
is given by $A_q(n,k)A_q(m,k)\big/{A_q(k,k)}$. An alternative way to express $A_q(n,k)$ is to use the quantity $H_q(n) = \prod_{i=1}^n (1-q^{-i})$ \cite{Cooper00}.
It can be proved that the limit of $H_q(n)$ for $n \to \infty$ converges to $H_q = \lim_{n \to \infty} H_q(n) \leq 1$ and $\lim_{q \to \infty} H_q = 1.$ 
Notice that $H_q(n) \leq H_q(n-k)$ and $0.28 \approx H_2\leq H_q < H_q(n) < 1$. 
Using some simple algebraic manipulation one can show that 
\begin{equation}\label{eq:Aq->Hq}
	A_q(n,k) = q^{nk} H_q(n)/H_q(n-k).
\end{equation} This expression reflects the asymptotic value of $A_q(n,k)$.
When we want to estimate the proportion of the matrices of rank $k$ among all the matrices in $\Fq^{k \times n},$ we have 
   \begin{equation}\label{Eq_prob_matrix}
   \begin{split}
        & A_q(n,k)/q^{nk} = 
    H_q(n)/H_q(n-k) \\ 
    & =  \prod_{i= n-k + 1 }^n(1-q^{-i}) 
      = \prod_{i= 0 }^{k-1}(1-q^{(i-n)})
    \\&\approx
    1 - q^{-n} \frac{q^k - 1}{q-1} \geq 1 - \frac{q^{-(n-k)}}{q-1}.
   \end{split}
   \end{equation}
The error probability of decoding the proposed BD-LRPC codes will be derived from 
 complicated calculations with the help of Ferrers diagrams. For reader's convenience, in the Appendix we recall relevant results about Ferrers diagrams, which will be heavily used in Section \ref{sec:ExpRowSub}.

\subsection{Rank-metric codes and their list decodability}

Rank-metric codes have both matrix and vectorial representations, for which there exists a one-to-one correspondence with respect to a basis of $\Fqm$ over $\Fq$.
Below we introduce some basics of rank-metric codes in vectorial representation.

A rank-metric code $\calC$ is simply a subset of $\Fqm^n$ equipped with rank metric, namely, its minimum rank distance is given by
$$
d(\calC):=\min\{d_R(\bx, \by)\mid \bx, \by \in \calC, \bx \neq \by\}
$$ with $d_R(\bx, \by) = \dim(\Span{\Fq}{\bx-\by}).$
Similarly to block codes with Hamming metric, there exist the
Singelton bound and Gilbert-Varshamov bound on rank-metric codes.
Given a rank-metric code $\calC\subseteq \Fqm^n$ with 
minimum rank distance $d$, its size satisfies
$$
\log_q|\calC| \leq \min \{n(m-d+1), m(n-d+1)\}.
$$
Without loss of generality, we may assume that $n\leq m$ since 
the minimum rank distance and size of a rank-metric code are preserved under transpose. Let $R=\frac{\log_{q}|\calC|}{mn}$ be the code rate and $\delta=\frac{d(\calC)-1}{n}$ be the relative distance. The asymptotic version of the Singleton bound, as $n\rightarrow \infty$, is given by
$R\leq 1- \delta$. 

Given a rank-metric code $\calC$, a received word $\by$ containing an error $\be$ of rank weight $r\leq \lfloor\frac{d(\calC)-1}{2}\rfloor$ can be uniquely decoded, indicating that there exists a unique codeword $\mathbf{c} \in \calC$ satisfying 
$d_R(\mathbf{c}, \by)=r$.
On the other hand, when the error $\be$ in $\by$ has rank weight $r> \lfloor\frac{d(\calC)-1}{2}\rfloor$, there might be a list of codewords having rank distance $r$ from $\by$. In this case one needs to turn to the list decoding.
However, list decoding for rank-metric codes appears more challenging than its counterpart in Hamming metric \cite{WachterZeh,Yang2015,Xing-Yuan2018}. Below we recall some results of list decoding for rank-metric codes.

A rank-metric code $\calC\subseteq \Fqm^n$ is said to be 
$(\rho, L)$-list decodable if for every $\bx\in\Fqm^n$, we have 
$$
\left|
\calC \cap B(\bx, \rho n)
\right| \leq L,
$$ where $B(\bx, r):= \{\by \in \Fqm^n \mid d_R(\bx, \by)\leq r\}$ is the rank metric ball of center $\bx$ with radius $r$.
It is clear that for decoding radius $ r=\rho n  \leq \lfloor\frac{d(\calC)-1}{2}\rfloor,$  any rank-metric code is $(\rho, 1)$ list-decodable, which corresponds to the case of unique decoding radius.
Ding in \cite{Yang2015} studied the theoretic limit of the list decodability of rank-metric codes.

\begin{proposition}\label{th:singleton-bound}\cite[Th. 1]{Yang2015}
Let $m$, $n$, $L$ be positive integers such that $b=n/m$ is a constant and $L=O(poly(mn))$. Then, for any $R\in (0,1)$, $\rho \in (0,1)$, a $(\rho, L)$-list decodable rank-metric code $\calC\subseteq \Fqm^n$ of rate $R$ must satisfy
$$
R\leq (1-\rho)(1-b\rho).
$$
\end{proposition}
It is also shown \cite{Yang2015} that for any $\epsilon\in(0,1)$,  a random $\Fq$-linear rank-metric code is  $(\rho, O(\exp(\frac{1}{\epsilon})))$-list decodable with rate $R=(1-\rho)(1-b\rho)-\epsilon$, which 
is considered as the Gilbert–Varshamov bound for the list decoding of rank-metric codes. However, there have been only a few efficient list decodable rank-metric codes in the literature \cite{Silva_2010, Guruswami-Wang-Xing2014, Xing-Yuan2018}. 

The rank-metric code in \cite{Silva_2010}
was proposed to achieve the capacity of an additive matrix channel and was given in matrix form 
$$
C=\begin{pmatrix}
    0_{v \times v} & 0_{v \times(n-v)} \\
0_{(m-v) \times v} & U
\end{pmatrix},
$$ where $U\in \mathbb{F}_q^{(m-v) \times(n-v)}$
is a data matrix \footnote{Here we transpose the code matrix in \cite{Silva_2010} for the consistency of notation in this paper.}. The authors in \cite{Silva_2010}
showed that taking $r=v$, this simple code has code rate $R=\frac{(n-r)(m-r)}{mn}$ in $q$-ary units and that 
for any $\epsilon \in (0,1)$, $r=(\rho+\epsilon)n$ and a constant $b=n/m$, as $m$ increases, this simple code can be uniquely decoded with overall complexity $O(nmr)$ operations in $\Fq$, the failure rate 
decreasing exponentially with $n$,  and the code rate satisfying
$
R\geq (1-\rho)(1-b\rho) - (1+b)\epsilon.
$ A modified version of this code was also discussed in the context of code-based cryptography \cite{gaborit2017}.
In \cite{Guruswami-Wang-Xing2014} 
an 
$\Fq$-linear subcode
of Gabidulin code was shown to be 
list-decodable with code rate $(1-\epsilon)R$, decoding radius $\rho=\frac{s}{s+1}(1-R)$, list size $O(q^{s^2/\epsilon^2})$ for any $\epsilon \in (0,1)$ and integer $s\geq 1$. Note that as $s$ increases, the relative decoding radius $\rho$ approaches the theoretic limit $1-R$. Nevertheless, 
the column-to-row ratio $b=\frac{n}{m}$ for the code is required to approach zero. Xing and Yuan recently in \cite{Xing-Yuan2018} reviewed the progress of list decoding of rank-metric codes and proposed an open problem:
{\it for a given constant ratio $b=\frac{n}{m}$, explicitly construct rank-metric codes of rate $R$ 
in $\Fq^{m\times n}$ with list decoding radius $\rho>\frac{1}{2}(1-R)$ and efficiently decode them. 
}
They provided an explicit construction of rank-metric codes, for which 
they provided both deterministic and probabilistic algorithms for decoding the codes. 
They showed that the constructed code 
for a constant $b$ up to $1/2$ and any $\epsilon\in (0,1)$ is $(\rho, \exp(O(1/\epsilon^2)))$-list decodable with $\rho$ beyond $\frac{1}{2}(1-R)$ and 
complexity $\text{poly}(n, \exp{(1/\epsilon)})$.

\smallskip

In subsequent sections, we will study a subfamily of $\Fqm$-linear low-rank-metric codes \cite{gaborit2013} and show that the code can be, with a high probability, 
uniquely decoded with radius $\rho = 1 - R - \frac{u}{n}$ for a constant $b=\frac{n}{m}$ and a positive integer $u$. The proposed codes for properly chosen parameters serve as a probabilistic solution to the open problem in \cite{Xing-Yuan2018}.

\subsection{Low-Rank Parity-Check (LRPC) Codes}
LRPC codes were introduced in 2013 by Gaborit, Murat, Ruatta and  Zémor \cite{gaborit2013}. Since then, they have been used in many cryptographic schemes \cite{ROLLO, Aragon2019} owing to their weak algebraic structure and efficient decoding.
\begin{definition}
	An $\Fqm$-linear vector code $\calC \subseteq \Fqm^n$ of $\Fqm$-dimension $k$ is said to be an LRPC code of rank density $d$ if it admits a parity-check matrix $H \in \Fqm^{(n-k) \times n}$ whose support has dimension $\dim(\Span{\Fq}{H}) = d$. 
\end{definition}
Given $\calH \subseteq \Fqm$ an $\Fq$-linear subspace of dimension $d$ and $H \in \calH^{(n-k) \times n}$ of rank $n-k$,
in most cases we have $\Span{\Fq}{H} = \calH$. The code $\calC$ having $H$ as a parity-check matrix is an LRPC code of density $d$.
We recall below the standard decoding of LRPC codes and its improved variant \cite{Aragon2019}.

Assume $\be \in \Fqm^n$ is an error of rank weight $r$. 
Let $\Span{\Fq}{H}=\Span{\Fq}{h_1,\ldots,h_d}$, and $\calE$ 
be the support of error given by $\Span{\Fq}{\be} = \Span{\Fq}{\varepsilon_1,\ldots,\varepsilon_r}$. 
Then the syndrome $\bs = \be H^\intercal$ has support $\calS = \Span{\Fq}{\bs}$ contained in the product space $\Span{\Fq}{H}.\calE.$
With properly chosen parameters, one can assume that $\calS=\Span{\Fq}{H}.\calE$ holds with a high probability, and
one can recover $\calE$ from $\calS$ and $\Span{\Fq}{H}$ under the assumption $$\calE =\bigcap_{i=1}^d h_i^{-1} \calS = \bigcap_{i=1}^d h_i^{-1} (\Span{\Fq}{H}.\calE).$$
From the knowledge of the error support, it is relatively easy to expand the linear system over $\Fqm$ (given by the $n-k$ coordinates of the syndrome $\bs$) over $\Fq$.
Usually this gives $(n-k)rd$ linearly independent equations in $nr$ variables in $\Fq$ which can be uniquely solved when $(n-k)d \ge n$.

However, this approach does not work for the case where $\calS \subsetneq \Span{\Fq}{H}.\calE$. This case was further discussed in \cite{Aragon2019}.
\begin{remark}\label{rem_1} With the above notation, when $\calS \subsetneq \Span{\Fq}{H}.\calE$,  the authors of \cite{Aragon2019} noticed that $\calE \cap h_i^{-1} \calS$ has dimension at least $r-c$ where $c$ is the co-dimension of $\calS$ in $\calH.\calE$.
	This means that the space $h_j h_i^{-1}\calS$ might contain elements of $\calH.\calE$ which were not present in $\calS.$
	Based on this observation they proposed to expand the syndrome space $\calS$ using the function
	$$
	f_{\mathrm{expand}}: \calS \rightarrow (\calS + h_j h_i^{-1}\calS )\cap(\calS + h_l h_k^{-1}\calS).
	$$
	This expansion can be iterated using different $i,j,k,l$ until either $\calS = \calH.\calE$ or all the possible expansions are tried.
	
	In \cite[Theorem 5.2]{Aragon2019} they analyzed the case $d=2.$
	To ensure a recovery of the space $\calE.\calH$ it was required that $\calS + h_1 h_2^{-1} \calS$ has dimension $3r$.
	This leads to the bound $3r \leq 2(n-k)$, implying $R \leq 1- \frac{3}{2} \rho.$
	In the same paper \cite[Section 5.4]{Aragon2019} it was also discussed how, for very large values of $m$, one may improve the expanding functions.
	However, those functions cannot be defined for $d=2.$
\end{remark}
\section{Bounded-degree LRPC codes}\label{sec:bdlrpc}
In this section we will consider a particular family of subspaces in $\Fqm$. The properties of these subspaces motivate us to study the BD-LRPC codes.

\begin{definition}\label{Def_Vd}
    Let $\alpha \in \Fqm$ and $d$ be a positive integer. The \textbf{bounded-degree subspace} generated by $\alpha$ of degree $d$ is described as 
    $
        \calV_{\alpha,d} = \Span{\Fq}{1, \alpha, \dots, \alpha^{d-1}}. 
    $
\end{definition}
From the definition, if $\alpha \in \Fqm$ does not belong to any proper subfield of $\Fqm$, we immediately have $\calV_{\alpha,m} = \Fqm$.
The product of two bounded-degree subspaces $\calV_{\alpha,i}$ and $\calV_{\alpha,j}$ has the following properties.
\begin{lemma}\label{lm:V+aV}
    For $\alpha \in \Fqm$ and $i,\,j\in [m]$, if $\alpha$ does not belong to any proper subfield of $\Fqm,$ we have $\dim(\calV_{\alpha,i}) = i$ and 
    \begin{equation}\label{eq:i.j=i+j-1}
         \calV_{\alpha,i}.\calV_{\alpha,j} = \calV_{\alpha,i+j-1},
            \end{equation}
    in particular, for $j=2$ we have
    \begin{equation} \label{eq:V+aV1}
        \calV_{\alpha,i+1} =  \calV_{\alpha,i}.\calV_{\alpha,2} = \calV_{\alpha,i} + \alpha \calV_{\alpha,i}.
    \end{equation}
\end{lemma}
\begin{proof}
    For any $i, j \in [m]$,
    a basis of $\calV_{\alpha,i}$ is given by $(1,\alpha,\ldots, \alpha^{i-1})$ and a basis of $\calV_{\alpha,j}$ is given by $(1,\alpha,\ldots, \alpha^{j-1}).$
    Furthermore,
    by the definition of $\calV_{\alpha,i}.\calV_{\alpha,j}$ we see that, if $i+j-2 \leq m$, $(1,\alpha,\ldots, \alpha^{i+j-2})$ is a basis of $\calV_{\alpha,i}.\calV_{\alpha,j}$, indicating that $ \calV_{\alpha,i}.\calV_{\alpha,j} = \calV_{\alpha,i+j-1}.$
    By taking $j=2$, one has $\calV_{\alpha,i}.\calV_{\alpha,2}= \calV_{\alpha,i+1}$, which has a basis $(1,\alpha,\ldots, \alpha^{i}).$ 
    Since $\calV_{\alpha,i} + \alpha \calV_{\alpha,i}$ has the same basis $(1,\alpha,\ldots, \alpha^{i})$, the statement follows.
\end{proof}
Motivated by decoding LRPC codes when $\calS \subsetneq \Span{\Fq}{H}.\calE$, 
we consider the following subfamily of LRPC codes. \cn
\begin{definition}\label{def:bd-lrpc}
    Let $\calV_{\alpha,d}=\Span{\Fq}{1,\alpha,\ldots, \alpha^{d-1}}$ be a bounded-degree subspace and
    $H \in \calV_{\alpha,d}^{(n-k)\times n}$ have full rank $n-k$.
    A code $\calC \subsetneq \Fqm^n$ having $H$ as its parity-check matrix is said to be a \textbf{bounded-degree low-rank parity-check (BD-LRPC)} code when $d$ is significantly smaller than $m$. 
\end{definition}

\begin{remark}\label{rm:LRPC_2}
    Observe that LRPC codes of density $2$ correspond exactly to BD-LRPC code of degree $2.$
    Let $\Span{\Fq}{a_1,a_2}$ be the support of  the parity-check matrix $H \in \Fqm^{(n-k) \times n}$ of an LRPC code of density $2$.
Then the matrix $H' = a_1^{-1} H$ will be a parity-check matrix of the same code and
$$
    \Span{\Fq}{H'} =  a_1^{-1}  \Span{\Fq}{H} = \Span{\Fq}{1, a_1^{-1} a_2} = \calV_{\alpha ,2},
$$
where $\alpha = a_1^{-1} a_2.$ In Figure \ref{fig:BD_LRPC-LRPC} we depict the situation of BD-LRPC of degree $d$ being a subset of LRPC codes of density $d$ for $d>2$ while being the same set for $d=2.$ 
\end{remark}
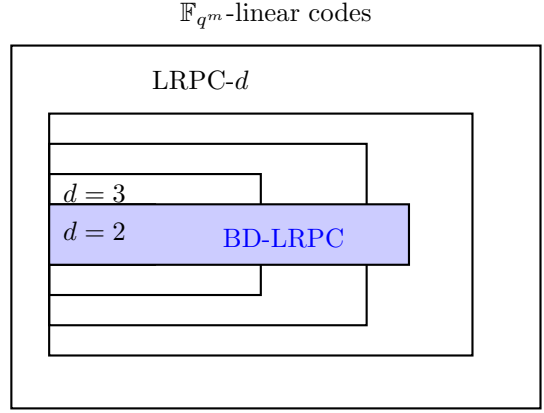
\begin{figure}
    \centering
    \begin{tikzpicture}  

    \foreach \r in {0.5,1.0,1.5,2.0}{
      \draw[thick] (0,3 - \r*0.8) rectangle (3.5*\r*0.8, 3 + \r*0.8);
    }

    \draw[thick, fill=blue!20, opacity=0.5] (0,3-0.5*0.8) rectangle (3.5*1.7*0.8, 3 + 0.5*0.8);

    \draw[thick, black] (-0.5,0.7) rectangle (6.5,5.5);

    \node[above] at (3, 5.6) {$\Fqm$-linear codes};
    \node[above] at (2, 4.8) {LRPC-$d$};
    \node[above] at (0.6 , 2.8+0.5) {$d=3$};
    \node[above] at (0.6, 2.8) {$d=2$};
    \node[above, blue, thick] at (3.1, 2.7) {BD-LRPC};
\end{tikzpicture}
    \caption{BD-LRPC inside the set of LRPC codes}
    \label{fig:BD_LRPC-LRPC}
\end{figure}

\medskip

Below we discuss the decoding of the BD-LRPC codes, which utilizes the structure of the bounded-degree subspace to improve the decoding capability.
For BD-LRPC codes with $\Span{\Fq}{H} = \calV_{\alpha,d},$ as a direct consequence of Lemma \ref{lm:V+aV},  we have the following result regarding the product space of $ \calV_{\alpha,d}$ and $\calE.$
\begin{lemma}\label{lm:VjViE}
    Let $\alpha\in \Fqm$, $\calV_{\alpha,d}=\Span{\Fq}{1, \alpha, \dots, \alpha^{d-1}} \subseteq \Fqm$ and let $\calE \subseteq \Fqm$ be an $\Fq$-linear subspace of $\Fqm$ of dimension $r$.
    Then, for any integer $i\geq 1$, 
    \begin{equation} \label{eq:ViE_rec}
        \calV_{\alpha,i+1}.\calE = \calV_{\alpha, 2}.\calV_{\alpha,i}.\calE = \calV_{\alpha,i}.\calE + \alpha \calV_{\alpha,i}.\calE.
    \end{equation}
\end{lemma}
The dimension of $ \calV_{\alpha, i}.\calE$ is in most cases equal to $ir$ for integers $i \leq m/r$ and it is always upper bounded by $ir$.
Consider $\calS = \Span{\Fq}{s_1, \ldots, s_{n-k}} \subsetneq \calV_{\alpha,d}.\calE$ and suppose its dimension is $\dim(\calS) = n-k < \dim(\calV_{\alpha,d}.\calE) \leq rd.$
In this case, the standard decoding algorithm for LRPC codes would result in a decoding failure. 
The observation in Lemma \ref{lm:VjViE} motivates us to expand $\calS \subsetneq \calV_{\alpha,d}.\calE$ with $\calV_{\alpha, 2}$ as follows
\begin{equation}\label{eq:V2S}
    \calV_{\alpha,2}.\calS \subseteq \calV_{\alpha,2}.\calV_{\alpha,d}.\calE = \calV_{\alpha,d + 1}.\calE.
\end{equation}
Notice that $\calV_{\alpha,2}.\calS = \calS + \alpha \calS$ can have dimension up to $2(n-k)$ while $\dim(\calV_{\alpha,d + 1}.\calE) \leq (d+1)r.$
That is to say, starting from $\calS \subsetneq \calV_{\alpha,d}.\calE$, if one expands the subspaces on both sides by multiplying with $\calV_{\alpha,2}$, 
the dimension incremental from $\calS$ to $\calV_{\alpha, 2}.\calS$ on the left-hand side of (\ref{eq:V2S}) is probably larger than the dimension incremental of $\calV_{\alpha,d}.\calE$ to $\calV_{\alpha,d+1}.\calE$ on the right-hand side. 
In particular, when $2(n-k) \geq (d+1)r,$ we would have
$\calV_{\alpha,2}.\calS = \calV_{\alpha,d + 1}.\calE$.
When this equality does not hold, we can further proceed with the expansion.
Provided that $m$ is sufficiently larger than $r(d+t-1)$, we can continue expanding $\calS$ with $\calV_{\alpha,t}$ for $t\geq 2$, which gives
\begin{equation}\label{eq:expansion_at}
     \calV_{\alpha,t}.\calS \subseteq \calV_{\alpha,t}.\calV_{\alpha,d}.\calE = \calV_{\alpha,d +t-1}.\calE.
\end{equation}
In this case, we have $\dim(\calV_{\alpha,t}.\calS) \leq t(n-k)$ while $\dim(\calV_{\alpha,t+d-1}.\calE) \leq (d+t-1)r$. 
When $t(n-k) \geq (d+t-1)r$, if $\calV_{\alpha,t}.\calS = \calV_{\alpha,d + t - 1}.\calE$, we consider it a \textit{successful expansion} since the error support $\calE$ can be then recovered from $\calV_{\alpha,t}.\calS$.
For sufficiently large $m$, the successful expansion occurs with a high probability. 
Rigorously showing this fact requires complicated calculations and the process is lengthy. In Section \ref{sec:Successful} we will elaborate the process of transforming $\calV_{\alpha,t}.\calS = \calV_{\alpha,d + t - 1}.\calE$ into a counting problem,
and then intensively study the counting problem with the help of Ferrers diagrams in Section \ref{sec:ExpRowSub}. To summarize, it will be shown in Sections \ref{sec:Successful} and \ref{sec:ExpRowSub}  that a successful expansion occurs with a probability lower bounded by 
$1-{q^{-u+1}}$ for $t=(d-1)r$ and $u = n-k-r$.
\begin{remark} As recalled in Remark \ref{rem_1}, the idea of expanding subspace was proposed in \cite{Aragon2019} in decoding LRPC codes. 
    The case $d=2$ was particularly analyzed in \cite[Th. 5.2]{Aragon2019}. In this case the parity-check matrix has a support $\calH = \Span{\Fq}{h_1,h_2} $ of $\Fq$-dimension $2,$ which
    can be considered as a BD-LRPC code of degree $2$ with parity-check matrix support given by $\Span{\Fq}{H}=\Span{\Fq}{1,\alpha}$ where $\alpha = h_1h_2^{-1}.$
    To underline the similarities between the two methods we will stick to the notation $\Span{\Fq}{H}=\Span{\Fq}{1,\alpha} = \calV_{\alpha,2}.$
    It was noted that $\dim(\calE.\calV_{\alpha,2} + \alpha \calE.\calV_{\alpha,2}) \leq 3r.$
    Under the assumption that the equality holds, the authors of \cite{Aragon2019} applied the only possible expansion for $d=2$ which is given by  $(\calS + \alpha \calS) \cap (\calS + \alpha^{-1} \calS).$
    It was noticed that $\alpha(\calS + \alpha^{-1} \calS) = \calS + \alpha \calS$ which implies that the two spaces have the same dimension.
    To ensure a recovery of the space $\calE.\calV_{\alpha,2}$ it was required that $\calS + \alpha \calS$ has dimension $3r$.
    This leads to the bound $3r \leq 2(n-k)$ which means that the rate can be at most $R \leq 1- \frac{3}{2} \rho.$

    With the expansion presented in this paper, it is possible to go further than $\calS + \alpha \calS.$
    When we apply our expansion for $d=t=2,$ we get $\calS + \alpha \calS \subseteq \calE.\calV_{\alpha,3}$ from which we obtain the same condition $3r \leq 2(n-k).$
    For larger $t$, we can get the improved condition $ (d+t-1)r \leq t(n-k)$, allowing for codes with a better rate $R \leq 1 - \frac{d+t-1}{t}\rho.$
\end{remark}

\smallskip

With a successful expansion $\calV_{\alpha,t}.\calS = \calV_{\alpha,d + t - 1}.\calE\subsetneq \Fqm$, one can recover the error support $\calE$ from $\calV_{\alpha, t}.\calS$ instead of the original $\calS$ as in the decoding of LRPC codes. 
With properly chosen parameters, the procedure leads to a unique decoding for the BD-LRPC codes.  
Here we will assume $\dim(\calV_{\alpha,d+t-1}.\calE) = (d+t-1)r.$ 
Although this dimension might be smaller, this assumption represents both the most common and the worst case.
In fact, if $\dim(\calV_{\alpha,d+t-1}.\calE)< (d+t-1)r$ were smaller, we might generate $\calV_{\alpha,d+t-1}.\calE$ by expanding a smaller subspace $\calS$ (i.e., using a shorter syndrome) or through a smaller number of expansions. In the subsequent discussion, we will focus on the worst case.

The decoding of BD-LRPC codes constitutes three phases: 
\begin{itemize}
	\item \textbf{Phase 1.} Expand $\calS$ to derive a successful expansion $\calV_{\alpha,t}.\calS = \calV_{\alpha,d + t - 1}.\calE$;
	\item \textbf{Phase 2.} Recover $\calE$ by intersecting $\calV_{\alpha,t}.\calS$ and  $\alpha^{-(t+d-2)} \calV_{\alpha,t}.\calS$;
	\item \textbf{Phase 3.} Recover the error $\be\in \F_{q^m}^n$ from $\mathbf{s} =\be H^\intercal= \bee EH^\intercal$, where $\bee$ is an ordered basis of $\calE$.
\end{itemize} 
Phase 1 is critical for the decoding. The analysis for successful expansion is complicated.
For the reader's convenience,  the decoding phases and derived results in Sections \ref{sec:Successful} and \ref{sec:ExpRowSub} for Phase 1 are summarized in Figure \ref{fig: decoding}. 
\cn
\begin{figure*}[tb]
	\begin{center}
		\begin{tikzpicture}[scale=0.95] 
		\tikzset{
			rect/.style={rectangle, draw, align=center, minimum width=5.2cm},
			rect1/.style={rectangle, draw, align=center, minimum width=2cm, minimum height=0.8cm},
			rect2/.style={rectangle, draw, align=center,  minimum height=0.8cm},
			rect3/.style={rectangle, draw, align=center, minimum width=5.5cm, minimum height=0.8cm, fill=lightgray!40},
			loop/.style={ 
				draw,
				chamfered rectangle,
				chamfered rectangle xsep=3cm
			},
		}

		\node at (-2,4){\textbf{Phase 1}};
		\node[rect](rect1) at (0, 3) {Successful Expansion \\ $\calV_{\alpha, t}.\calS = \calV_{\alpha, d+t-1}.\calE$};
		
		\node at (-2,1){\textbf{Phase 2}};
		\node[rect](rect2) at (0, 0) {Error Support \\ $\calE = \calV_{\alpha,t}.\calS \cap  \alpha^{-(t+d-2)} \calV_{\alpha,t}.\calS$};
		
		\node at (-2,-2){\textbf{Phase 3}};
		\node[rect](rect3) at (0, -3) {Syndrome Equation \\ $\mathbf{s} = \bee EH^\intercal$};
		
		\node (rect0) at (0, 5){$\mathbf{s} = \mathbf{e} H^\intercal$};	  
		\node at (1.3, 4.2){$\calS = \Span{\F_q}{\bs}$};
		\node at (1.5, -1.5)	{a basis $\bee$ of $\calE$};
		\node at (4, 3.2)	{condition};
		
		\node at (0.4, -4){$E$};
		\node (rect4) at (0, -5){$\mathbf{e}=\bee E$};

		\draw[->] (rect0) -- (rect1);  
		\draw[->] (rect1) -- (rect2);
		\draw[->] (rect2) -- (rect3);  
		\draw[->] (rect3) -- (rect4);

		\node at (6.8,3.95){\text{Section \ref{sec:Successful}: Prop. \ref{prop:0}}};
		\node[rect3](rect5) at (8, 3) {$\Rank(X_1) =r$ \\ $\Rank(M_t(Z^{(1)},A)) = (d-1)r$ };

		\node at (6, 0.7) {Sec. \ref{subsec:d=2}};
        \node at (8.7, -0.7) {Sec. \ref{subsec:d>2}};  
		
		\node[rect1](rect6) at (6.5, 0) {Uniform $A$ };
		\node[rect1](rect8) at (8.8, 0) {$A$ in \eqref{eq:comp_mat}};
		\node[rect2](rect9) at (10.645, 0) {Tab. \ref{tab1}};

		\node at (6, -1) {Section \ref{sec:ExpRowSub} };
		\node[rect3](rect7) at (8, -2) {Problem \ref{Pb:C_k}\\ Th. 3-4 and Cor. 1-2  };
		
		\node at (6,-3.2){Lemma \ref{lm:SyndromeUniform}};
		\node[rect](rect5a) at (8, -4) {$X_1,\dots, X_d$ are uniform };
		
		\draw[->] (rect6) -- (6.5, 2.35);  \node at (7.1, 1.5){$d=2$}; 	\node at (9.8,1.5) {$d>2$};
		\draw[->] (rect8) -- (8.8, 2.35);  \node at (10.25, 0.75){$ r=1$ \quad\, $r>1$};
		\draw[->](10.55,0.4) -- (10.55, 2.35);
		
		\draw[->] (7.2, -1.4) -- (7.2, -0.44);   
		\draw[->] (rect5a) -- (rect7);  
		
		\draw[<->] (rect5) -- (rect1);
		\end{tikzpicture}
	\end{center}
	\caption{Decoding of BD-LRPC Codes: Phases and Auxiliary Results}\label{fig: decoding}
\end{figure*}
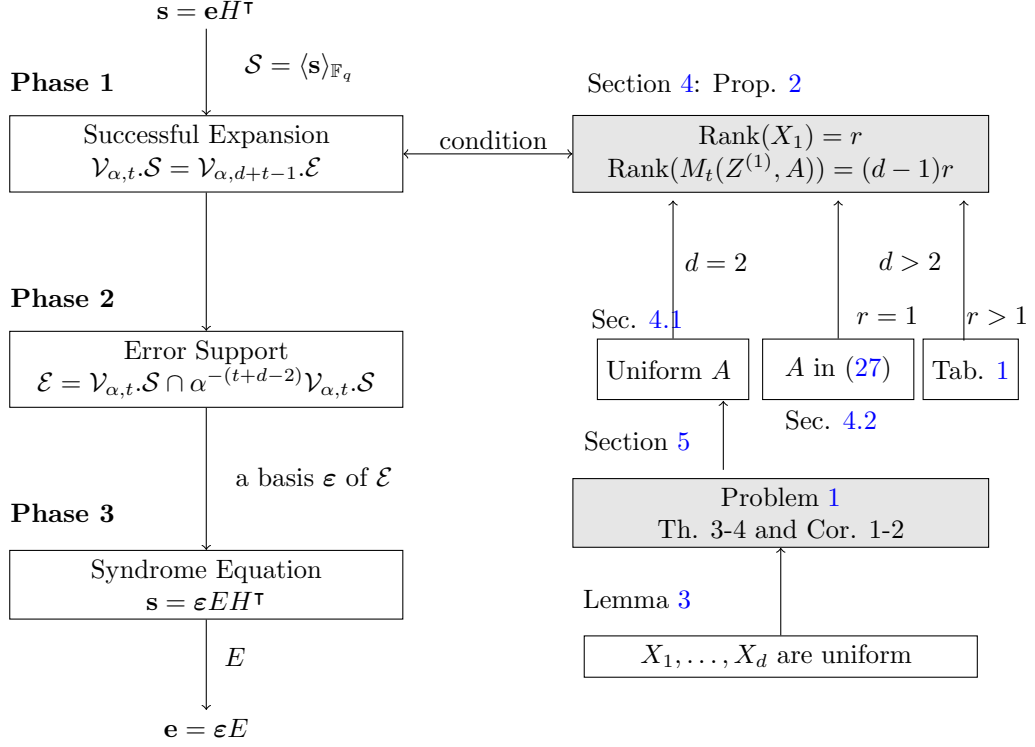

\medskip  

We need to discuss the necessary conditions on the parameters that affect the decoding procedure for BD-LRPC codes. 

\textbf{Phase 1.} Suppose $\calV_{\alpha,t}.\calS = \calV_{\alpha,d + t - 1}.\calE$
for certain positive integer $t$.
Such a successful expansion on $\calS$ requires, in terms of dimension, the inequality $(n-k)t \geq (d+t-1)r$.
Denoting $u = (n-k)-r$, we have $u \geq \frac{d-1}{t}r$ and a lower bound on $t$ given by $t \geq \frac{d-1}{u}r$.
In Section \ref{sec:ExpRowSub} \cn, we will see that 
if $\calV_{\alpha,t}.\calS = \calV_{\alpha,d+t-1}.\calE$ does not hold for $t$ up to $(d-1)r$,
the equality will no longer be reachable, indicating that it suffices to proceed with expansion on $\calS$ for $t$ up to $(d-1)r$. 
Hence the bounds on the positive integer $t$ for a successful expansion $\calV_{\alpha,t}.\calS = \calV_{\alpha,d+t-1}.\calE$ are given by  
\begin{equation}\label{eq:t_bounds}
    \frac{d-1}{u}r \leq t \leq (d-1)r. 
\end{equation}

\textbf{Phase 2.} 
 To get $\calE$ from a successful expansion,  
we observe that 
\begin{align*}
    \calV_{\alpha,d+t-1}.\calE  &=  \calE + \alpha \calE + \ldots +  \alpha^{d+t-2}\calE \\
    \alpha^{-1}(\calV_{\alpha,d+t-1}.\calE)  &=  \alpha^{-1}\calE + \calE  + \ldots + \alpha^{d+t-3} \calE.
\end{align*}
This means that $ \calV_{\alpha,d+t-2}.\calE \subseteq \calV_{\alpha,d+t-1}.\calE \cap \alpha^{-1}(\calV_{\alpha,d+t-1}.\calE).$ For a similar reason we have that $\calV_{\alpha,d+t-i-1}.\calE \subseteq \calV_{\alpha,d+t-1}.\calE \cap \alpha^{-i}(\calV_{\alpha,d+t-1}.\calE).$
Since we want the intersection to be just $\calE,$ the only meaningful intersection is  between the two spaces
\begin{align*}
    \calV_{\alpha,d+t-1}.\calE  &=  \calE + \alpha \calE + \cdots +  \alpha^{d+t-2}\calE, \\
   \alpha^{-(t+d-2)}(\calV_{\alpha,d+t-1}.\calE)  &=  \alpha^{-(t+d-2)}\calE  + \cdots + \calE.
\end{align*}
Hence, from a successful expansion it follows that
\begin{align*}
\calE  & \subseteq 
\calV_{\alpha,d+t-1}.\calE \cap  \alpha^{-(t+d-2)}(\calV_{\alpha,d+t-1}.\calE)   \\
  & = \calV_{\alpha,t}.\calS \cap  \alpha^{-(t+d-2)} \calV_{\alpha,t}.\calS.
\end{align*}
\cn
Denoting these two subspaces as $\calA, \, \calB \subsetneq \Fqm$, respectively, 
we have  that $\calE \subseteq \calA, \calE \subseteq \calB$ with $\dim(\calE) = r $ and $\dim(\calA) = \dim(\calB) = (d+t-1)r$.
In order to have $\calE = \calA \cap \calB$ we need that $\dim(\calA \cap \calB) \leq r = \dim(\calE).$
This gives us a necessary condition on $m.$ 
We start form the inequality
$$
m \geq \dim(\calA + \calB)
= \dim(\calA) + \dim(\calB) - \dim(\calA \cap \calB),
$$
this leads to 
$$
r \geq \dim(\calA \cap \calB) \geq \dim(\calA) + \dim(\calB) - m = 2(d+t-1)r - m
$$
thereby
\begin{equation}\label{eq:m>=2(d+t)r-3r}
m\geq 2(d+t-1)r - r.    
\end{equation} 

According to \cite[Prop.\,2.4.2]{ROLLO}, the probability that the intersection contains only $\calE$ can be estimated as $1 - q^{-(m - 2(d+t-1)r + r)},$
which is the approximate probability that two random subspaces $\calA, \calB \subseteq \Fqm$ of dimension $(d+t-1)r$ both containing the same subspace $\calE$ of dimension $r$ intersect only in $\calE$.

\cn

\textbf{Phase 3.} Once we recover the error support $\calE$ from the equality $\calV_{\alpha,t}.\calS = \calV_{\alpha,d+t-1}.\calE$, we can expand the syndrome equations $H\be^\intercal = \bs^\intercal$ over $\Fq$ in the same way as for the standard LRPC codes  \cite{gaborit2013,Aragon2019}.
Similarly, to have a unique solution, we need the resulting system of $(n-k)rd$ equations and $nr$ variables in $\F_q$ to be over-determined. This can happen only if $(n-k)d \geq n$.
To the best of our knowledge there is no analysis on the likelihood that this system leads to a unique solution, a sufficient condition is that the resulting linear system has rank $nr$.
In ROLLO \cite{ROLLO} and other cryptographic applications the parameters are usually chosen in such a way that there are way more equations than variables. In those cases, it seems plausible to assume that the system is indeed of maximal rank and has a unique solution.
A possible heuristic would be to consider the probability that an $(n-k)rd \times nr$ matrix over $\Fq$ has rank $nr$.
The probability is $0$ if $n > (n-k)d$ and can be lower bounded by $1-q^{-(n(d-1) -kd)r + 1}$ for $n <  (n-k)d.$
When $n=(n-k)d$ the probability could be estimated as $H_q(n) \approx H_q $ which is the probability that a random $n \times n$  matrix over $\F_q$ has full rank.

\smallskip

Based on the assumption $\dim(\calV_{\alpha,d + t - 1}.\calE) = (d+t-1)r$ and the  probability of $\calV_{\alpha,t}.\calS = \calV_{\alpha,d + t - 1}.\calE$ discussed in Sections \ref{sec:Successful}, \ref{sec:ExpRowSub} and the analysis in Phases 1-3, we obtain the following result.

\begin{theorem}\label{Th0}    
    Let $\calC \subseteq \Fqm^n$ be a BD-LRPC code of dimension $k$ and degree $d$, where $\frac{n}{n-k}\leq d \ll m$. 
	Then, 
	\begin{enumerate}[label=(\roman*)]
		\item for $d =2$, the code $\calC$ can uniquely decode errors of rank $r$ with a probability in the order of  $1- c_q q^{-u+1} ,$  
		with $u= n-k-r$, when
		$$
		u \leq \min \{m - 2r^2-r, (n-2k)r-1\};
		$$
		\item for $d>2$, the code $\calC$ can uniquely decode errors of rank $1$ with a probability in the order of  $1- c_q q^{-u+1} $ 
		with $u= n-k-1$ when $$
		u \leq \min \{m - 4d+5, n(d-1) -kd-1\}.
		$$
	\end{enumerate}
  where $c_q = \frac{q+1}{q-1}$ lies in the interval $(1,3]$ with $c_2=3$.
\end{theorem} \cn 
\begin{proof}
    We will discuss the unique decodability of the BD-LRPC codes based on 
    a successful expansion $\calV_{\alpha,t}.\calS = \calV_{\alpha,d+t-1}.\calE$. In Section \ref{sec:Successful},
    Proposition \ref{prop:0} points out that a successful expansion occurs if and only if 
   $$\Rank(X_1)=r \text{ and }\Rank(M_t(Z^{(1)}, A)) =(d-1)r,$$ where $X_1$  is the $(n-k)\times r$ matrix given in \eqref{eq:SyndromeSum}
   derived from the syndrome $\bs=(\bs_1,\dots, \bs_{n-k})$ with $\bs_i\in \calS \subsetneq\calV_{\alpha,d}.\calE$,
   the matrices $A$, $Z^{(1)}$
   are given in and after \eqref{eq:matrix_A}, respectively,
   and $M_t(Z^{(1)}, A)$ is the $ut \times (d-1)r$ matrix over $\Fq$ defined in \eqref{eq:MtZA}, 
    which is derived by applying elementary operations to the $(n-k)t\times (d+t-1)r$ matrix $M_t$ on the right-hand side of \eqref{eq:ExpAsMatrix} where $u = n-k-r$. Some of these matrices are given below for convenience:
    \begin{align*}
   & A =
   \begin{pmatrix}
   A_1 & A_2 & \cdots & A_{d-2} & A_{d-1} 	\\
   I_r &     &        &         &  \bzero  \\
   & I_r &  	   &         &  \bzero  \\
   &	  & \ddots &         &  \vdots  \\
   &     &		   & I_r     & 	\bzero	
   \end{pmatrix},
   	& M_t(Z^{(1)},A) = \begin{pmatrix}
   Z^{(1)} \\ Z^{(1)}A \\ \vdots \\ Z^{(1)}A^{t-1}
   \end{pmatrix}.
   \end{align*}
   We now discuss the probability of the successful expansion based on the results from Sections \ref{sec:Successful} and \ref{sec:ExpRowSub}.
 	From the process of obtaining $X_1$ from the syndrome equations, Lemma \ref{lm:SyndromeUniform} in Section \ref{sec:Successful} implies that $X_1$ can be seen as a (uniformly) random matrix in $\Fq^{(n-k)\times r}$. Consequently, the matrices $A_i$'s can be seen as randomly selected from $\F_q^{r\times r}$. 
     
     We first consider the probability of $\Rank(X_1)=r$.  According to \eqref{Eq_prob_matrix} in Section \ref{sec:prel},  the probability that a random matrix $X_1$ uniformly drawn from $\Fq^{(n-k)\times r}$, where $n-k = r + u$, has rank $r$
  is given by  
  $$
    \prod_{i = 0}^{r-1}(1- q^{i-(r+u)}) \approx 1 - q^{-r-u} \sum_{i=0}^{r-1}q^i  \approx 1 - \frac{q^{-u}}{(q-1)}.
  $$
   
	Regarding the probability of $\Rank(M_t(Z^{(1)}, A)) = (d-1)r$ for $t=(d-1)r$, we shall divide it into two cases:  $d=2$ and $d>2$. 
	
   (i) For $d=2$, the matrix $A$ reduces to $A_1$. Subsection \ref{subsec:d=2} shows that the matrix $\binom{A}{Z^{(1)}}$ can be deemed as a random matrix uniformly sampled from $\F_q^{(r+u)\times r}$.
   Section \ref{sec:ExpRowSub} intensively discusses the number of pairs of matrices $Z^{(1)}$ and $A$ such that the matrix $M_t(Z^{(1)}, A)$ has full rank.
    Corollary \ref{cor:Pr(C_k^u)} in Section \ref{sec:ExpRowSub} explicitly shows that for uniformly random matrices $Z\in \Fq^{u\times r}$ and $A\in \Fq^{r \times r}$
    \begin{align*}
        {\rm Prob}(\Rank(M_t(Z, A)) = r) & = \frac{H_q(r+u-1)}{H_q(u-1)} \\
        & \geq
        1-\frac{q}{q-1}q^{-u}.
    \end{align*}
    
    (ii) When $d>2$ and $r=1$, Subsection \ref{subsec:d>2} shows that, for the matrix $M_t(Z^{(1)}, A),$ achieving full rank $(d-1)r=(d-1)$ is equivalent to the conditions that the related polynomials $z_i(x)=\sum_{j=0}^{d-2}z_{i,j}x^{j}$, where $z_{i,j}$ is the $(i,j)$-th entry of $Z$, and $p_A(x)=x^{d-1} - \sum_{j=0}^{d-2}a_jx^j$ satisfy $$\gcd(p_A(x), z_1(x),\dots, z_u(x))=1.$$ Using a known result from \cite{Bennett07}, the probability of this second event is equal to $1-q^{-u}$, indicating that
    \[
    {\rm Prob}(\Rank(M_t(Z^{(1)}, A)) = (d-1)) = 1-q^{-u}. 
    \]\cn
    Combining the probabilities of $\Rank(X_1)=r$ and $\Rank(M_t(Z^{(1)}, A)) = (d-1)r$ for Cases (i) and (ii), the successful expansion $\calV_{\alpha,t}.\calS = \calV_{\alpha,d+t-1}.\calE$ holds with a probability lower bounded by $1-c_q q^{-u+1}$ with $c_q=\frac{q+1}{q-1}$.

\smallskip
    
     From the discussion before Theorem \ref{Th0}, the success probability for Phases 2 and 3 can be approximated by \cn
    \begin{equation}
	\begin{split}
	    &(1 - q^{-(m - 2(d-1)(r+1)r + r)})(1-q^{-(n(d-1) -kd)r + 1}) \\
        \approx & 1 - q^{-(m - 2(d-1)(r+1)r + r)} -  q^{-(n(d-1) -kd)r + 1}.
	\end{split}
    \end{equation}
	Therefore, the unique decoding of BD-LRPC codes has a success probability given by $1- c_q q^{-u+1}$ when
	$$
        u\leq  \min\{ m - 2(d-1)(r+1)r + r, ((d-1)n  - kd)r-1 \}.
    $$
    The above condition yields the corresponding conditions on $u$ for Cases (i) and (ii), which can hold easily for properly chosen $m$ and $n$.
    \cn
\end{proof}

It is worth noting that, according to our experimental results, the statement in Theorem \ref{Th0} appears to be true also for general cases where $d>2$ and $r>1$.

\smallskip

\textbf{Experimental results.}
%
We ran some experiments with MAGMA\footnote{Source code available at \url{https://chunlei-li.github.io/projects/bd_lrpc.html}}, where we generated the matrix 
$A$
 and a random matrix $Z \in \Fq^{u \times (d-1)r}$ with different sets parameters.
For each set of parameters we run $10000$ tests and count how many of these tests gave us $\Rank(M_t(Z,A)) = (d-1)r$ for $t=  (d-1)r.$ 
The experimental results are listed in Table \ref{tab1}.
The first column lists the values in the case that $Z,A$ were both chosen from a uniform distribution for $d=2$ in Theorem \ref{Th0} (i); the second column lists the values where $A$ is as in \eqref{eq:matrix_A} for $d>2, r>1$; and the third column lists the values where $A$ is a companion matrix as in \eqref{eq:comp_mat} for $d>2, r=1$ in Theorem \ref{Th0} (ii). 

It can be verified that the values in Column 1 are consistent with the result given in Corollary \ref{cor:Pr(C_k^u)} with the probability $H_q(r(d-1) + u -1)/H_q(u-1)$; and those in Column 3 match the probability $1 - q^{-u}$ for the case $d>2, r=1$.
Notice that $r=1$ in the first example corresponds to a companion matrix, while $d=2$ in the $5^{th}$ and $6^{th}$ examples correspond to the uniform matrix case discussed in Section \ref{sec:ExpRowSub}.

\begin{table}[h]
	\begin{center}
		\begin{tabular}{|c|c|c|c|}
			\hline
			$(q,r,u,d)$ & Uniform & Measured & Companion \\
			\hline
			$(2,1,2,5)$ & $5960$  & $7496$   & $7500$  \\
			$(2,2,2,5)$ & $5787$  & $6532$   & $7500$  \\
			$(2,3,2,5)$ & $5776$  & $6093$   & $7500$  \\
			$(2,4,2,5)$ & $5775$  & $5985$   & $7500$  \\
			$(2,7,2,2)$ & $5798$  & $5872$   & $7500$  \\    
			$(3,3,3,2)$ & $9471$  & $9469$   & $9629$  \\
			$(3,3,3,4)$ & $9452$  & $9491$   & $9629$  \\
			$(7,2,3,3)$ & $9966$  & $9966$   & $9970$  \\
			\hline
		\end{tabular}
	\end{center}
	\caption{Experiments for $\Rank(M_t(Z,A)) = (d-1)r$ in 10000 trials}\label{tab1}
	\label{tab:my_label}
\end{table}

However, providing rigorous proof for those general cases seems intractable. 
Neither of the techniques for proving the cases $d=2$ and $d>2,\, r=1$ can be applied to those general cases.
Below we propose a conjecture and cordially invite interested readers to attack it, which will lead to a complete work for BD-LRPC codes for any $d\geq 2$ and $r\geq 1$.

\begin{conj}\label{con_d>2} Let $Z$ be a uniformly random matrix in $\F_q^{u \times (d-1)r}$, $A\in \F_q^{(d-1)r\times (d-1)r}$ and $M_t(Z,A) $ be the matrices as in the proof of Theorem \ref{Th0}, where 
	$A_1,A_2,\dots, A_{d-1}$ are independently and uniformly chosen at random from $\F_q^{r\times r}$. 
    Then for $t=(d-1)r$, and $d>2,\, r\geq 2$, one has
    \begin{align*}
         &1-\frac{q^{-u+1}}{q-1} \approx \frac{H_q((d-1)r+ u -1)}{H_q(u-1)} \\
         &\leq  {\rm Prob}(\Rank(M_t(Z,A)) = (d-1)r) \leq 1-q^{-u}.
    \end{align*}

\end{conj}

\begin{remark}
According to Theorem \ref{Th0}, for BD-LRPC codes with a code rate $R=\frac{k}{n}$, the unique decoding radius is given by $1-R-\epsilon$ with $\epsilon = \frac{u}{n}$ with a decoding failure rate $c_qq^{-u+1}$ for a positive integer $u\leq  \min\{m - 2(d-1)(r+1)r + r, ((d-1)n  - kd)r -1 \}$.
In this case, as $n$ approaches infinity, we see that $\epsilon = \frac{u}{n} \rightarrow 0$ for a fixed constant $u$, which requires  the column-to-row ratio $b=\frac{n}{m}\leq \frac{1}{2(d-1)(1-R-\epsilon)^2 n +(2d - 3)(1-R-\epsilon) + \epsilon} \rightarrow 0$ since $m \geq 2(d-1)(r+1)r - r + u$.
To the best of our knowledge, simple codes \cite{Silva_2010}, \cite{gaborit2017}, are the only instances of rank-metric codes that can allow efficient unique decoding of the code rate $R$ approaching 
the Singleton bound $1-\rho$. Theorem \ref{Th0} provides a family of $\Fqm$-linear rank-metric codes (i.e., LRPC codes of density $2$) that can be uniquely decoded and have a code rate approaching 
the Singleton bound.
Assuming Conjecture \ref{con_d>2} is valid, we can extend this family of $\Fqm$-linear rank-metric codes to the larger family of BD-LRPC codes.
\end{remark}
    
    In Theorem \ref{Th0} we discussed the decoding radius of the BD-LRPC codes for the case $t= (d-1)r$ based on Corollary \ref{cor:Pr(C_k^u)} in Section \ref{sec:ExpRowSub}. 
    Experimental results show that, for a smaller integer $t \geq \lceil \frac{(d-1)r}{u} \rceil + 1$, the probability of a successful expansion $\calV_{\alpha,t}.\calS = \calV_{\alpha, d+t-1}.\calE$ is close to that for $t= (d-1)r$ as well.
	In the context of decoding BD-LRPC codes, it is of more interest to consider a smaller integer $t$, which in turn allows for a smaller integer $m$ by the condition $m\geq (2(d+t)-3)r$ given by (\ref{eq:m>=2(d+t)r-3r}). \cn
    This motivates us to investigate the probability of a successful expansion for a small integer $t \approx \lceil \frac{(d-1)r}{u}\rceil +1 $ in the last part of Section \ref{sec:ExpRowSub}. 

\begin{remark}
    In Corollary \ref{prop:Prob_t} in Section \ref{sec:ExpRowSub} we show that, for $d=2$ the probability that $M_t(Z^{(1)}, A)$ has rank $r$ for $t = \lceil \frac{r}{u}\rceil  + 1$ is lower bounded by $1-\frac{q^{-u/2} + q^{-u+1}}{q-1}$.
    Combining the probability of $\Rank(X_1)=r$ as discussed in Theorem \ref{Th0}, we have 
    \begin{align*}
        &{\rm Prob}(\calV_{\alpha,t}.\calS = \calV_{\alpha,t}.\calE) \geq 1-\frac{q^{-u/2} + q^{-u+1}+q^{-u}}{q-1} \\  & \approx 1 - \frac{q^{-u/2}}{q-1},
    \end{align*}
    for
    $t \geq \lceil \frac{r}{u}\rceil + 1.$
    Similarly to Theorem \ref{Th0}, the BD-LRPC codes can be uniquely decoded with a probability $1-q^{-\frac{u}{2}}$ when
    $$
    \begin{array}{c}
         \frac{u}{2} \leq  \min\left\{ m -  \left(\frac{2r}{u} + 3\right)r, ((d-1)n - kd)r \right\},
    \end{array}
    $$
    where the first term comes from substituting the lower bound of $t$ in (\ref{eq:m>=2(d+t)r-3r}) for $d=2$.
    In order to have a negligible failure probability in Phase $2$ compared to the failure probability of the expansion, we need
    $$
    \begin{array}{c}
        m \geq \left(\frac{2r}{u}+ 3 \right)r + \lceil \frac{u}{2} \rceil.
         \end{array}
    $$  
    Below we shall discuss the column-to-row ratio $b=n/m$, for which we only need to focus on the lower bound on $m$ in terms of $r, u$.
    Substituting $\rho=\frac{r}{n}$ and $\epsilon = \frac{u}{n}$ gives
    \begin{equation}\label{eq:n/m_asymptotic_2}
    \begin{array}{c}
        m \geq \left(\frac{2 \rho}{\epsilon}+ 3 \right)\rho n + \frac{\epsilon n}{2}.
    \end{array}
    \end{equation} 
    From \eqref{eq:n/m_asymptotic_2} we see that 
    the column-to-row ratio 
    \begin{equation}\label{bBound}
        b=\frac{n}{m} \leq \frac{2\epsilon}{2(2 \rho + 3 \epsilon) \rho + \epsilon^2} = f_\rho(\epsilon), 
    \end{equation}
    which indicates that $b$ can be a constant $b \leq f_\rho(\epsilon)$ as $n$ approaches infinity.
    Moreover, a closer look at the function $f_\rho(\epsilon)$ shows that it takes largest value at $\epsilon = 2\rho$ and it is monotonically increasing over the interval $(0, 2\rho]$.
    Given a constant $b=\frac{n}{m}$ as $n\rightarrow \infty$ and a value $\rho\in (0,1)$, 
    the BD-LRPC codes can uniquely decode errors of rank $r=\rho n$ with high probability, given that its code rate is upper-bounded by 
    $$
        R \leq 1 - \rho - \epsilon.
    $$
    Furthermore, in order to have $\rho \geq \frac{1}{2}(1-R)$
    we need that $\epsilon \leq \frac{1}{2}(1-R)$ which results into
    $$
        \rho \geq \frac{1-R}{2} \geq \epsilon.
    $$ 
    Let $\epsilon_0 < \rho < 2\rho$, we can fix a value $b_0 = f_\rho(\epsilon_0).$
    For any value of $\epsilon$ in the interval $ [\epsilon_0, \rho]$ the inequality in \eqref{bBound} will be satisfied when $b \leq b_0.$

    The largest value of $b$ that can be considered is obtained when $\rho = \epsilon = \frac{1}{2}(1-R)$ and corresponds to $b = f_\rho(\rho)= \frac{2}{11 \rho}.$
    For this specific case, we have that $\epsilon$ must be equal to $\rho = \frac{1}{2}(1-R).$
    For lower values of $b$, $\epsilon$ can be chosen into a larger interval $[f_\rho^{-1}(b), \rho)$, which implies $\rho > \frac{1}{2}(1-R) $.

    That is to say, the BD-LRPC codes can have a decoding radius beyond $\frac{1}{2}(1-R)$ when we take a constant $b< \frac{4}{11(1-R)}
    $ and $\rho \geq \frac{1-R}{2} \geq \epsilon$ as $n$ approaches infinity. Hence these codes serve as another solution to the open problem proposed in \cite{Xing-Yuan2018}. 
    Similar results can be obtained for the case $d>2$ under the assumption that Conjecture \ref{con_d>2} holds.
\end{remark}

\color{black}

\section{Successful subspace expansion.}\label{sec:Successful}

For BD-LRPC codes, assume $\calS \subseteq \Fqm$ is the subspace generated by the syndrome $\bs=\be H^\intercal$ and $\calV_{\alpha,d}.\calE \subseteq \Fqm$ is the product space between the parity-check matrix support $\calV_{\alpha,d}$ and the error support $\calE.$
In the previous section, we briefly pointed out that we can 
expand $\calS$ and $\calV_{\alpha,d}.\calE$ by 
the product with $\calV_{\alpha,t}$, for a certain integer $t\geq 1$, 
which leads to 
$\calV_{\alpha,t}.\calS = \calV_{\alpha,d + t -1}.\calE$ with a high probability.
In this section, we shall intensively discuss the probability of such a successful expansion.

Consider the space $\calE = \Span{\Fq}{\be}$ with a basis ${\bm\varepsilon} = (\varepsilon_1, \ldots, \varepsilon_r)$ and
$\calS =\Span{\Fq}{\bs}$, where $\bs=(s_1,\dots, s_{n-k})$ is the syndrome. Since $\calS$ is a subspace of the product space $\calV_{\alpha,d}.\calE$, we have
$$
\calS \subseteq \calV_{\alpha,d}.\calE = \calE + \alpha \calE  + \ldots + \alpha^{d-1} \calE .   
$$
Thus  each $s_i \in \calS$ can be written as $s_i = \sum_{j=1}^d x_{i,j}\alpha^{j-1}$, where $x_{i,j}\in \calE$ for $j=1,\dots, d$.
Denote by $\ba_d$ the vector $\ba_d = (1,\alpha, \ldots, \alpha^{d-1})$, which is a basis of $\calV_{\alpha, d}$ when $\dim(\calV_{\alpha, d}) =d$\cn.  Then we have $s_i = (x_{i,1}, \ldots, x_{i,d}) \ba_d^\intercal.$
Consider the matrix $X \in \calE^{(n-k) \times d}$ whose $i$-th row is given by $(x_{i,1}, \ldots, x_{i,d})$, where $1\leq i\leq n-k$. This gives 
$$
\bs^\intercal
=
X \ba_d^\intercal = \bx_1^{\intercal} + \alpha \bx_2^{\intercal} + \dots + \alpha^{d-1} \bx_d^{\intercal},
$$ where
$\bx_j=(x_{1,j}, \dots, x_{n-k,j}) \cn \in \calE^{n-k}$ and $\bx_j^\intercal$ is the $j$-th column of $X$.
We can expand each column $\bx_j^\intercal$ to the matrix $X_j \in \Fq^{(n-k) \times r}$ such that 
$
X_j \bee^{\intercal} = \bx_j^\intercal.
$
The syndrome $\bs$ can then be obtained as 
\begin{equation}\label{eq:SyndromeSum}
\bs^\intercal
=
X_1 \bee^\intercal + \alpha X_2 \bee^\intercal + \ldots + \alpha^{d-1} X_{d} \bee^\intercal. 
\end{equation}
The vector $
\ba_d \otimes \bee = 
(\bee, \alpha \bee, \ldots, \alpha^{d-1} \bee) \in \Fqm^{dr},
$ is always a generator of $\calV_{\alpha,d}.\calE$
implying
$\calV_{\alpha,d}.\calE = \Span{\Fq}{\ba_d \otimes \bee}.$ 
When $\dim(\calV_{\alpha,d}.\calE) = rd,$ the vector $\ba_d \otimes \bee$ is also a basis. \cn Using these notations we have 
\begin{equation}\label{eq:SyndromeAsMatrix}
\bs^\intercal
=
\begin{pmatrix}
X_1,  \ldots,  X_d 
\end{pmatrix}
(\ba_d \otimes \bee)^\intercal.
\end{equation}
This equation describes the coordinates of the syndrome with respect to the generator $\ba_d \otimes \bee$ of $\calV_{\alpha,d}.\calE$.
When $\dim(\calS) = \dim(\calV_{\alpha,d}.\calE) = dr$, 
the matrix $(X_1,\ldots, X_d)$ induces a one-to-one correspondence between $\bs$ and $\ba_d \otimes \bee$, thus $\bee$ can be uniquely determined.
\cn

\begin{remark}
	All the discussion above can be applied in the same way for a generic LRPC code.
	Let $H \in \Fqm^{(n-k)\times n}$ be a parity-check matrix such that $ \calH = \Span{\Fq}{H} = \Span{\Fq}{h_1, \ldots, h_d}.$
	If we consider the ordered basis $\mathbf{h}_d = (h_1, \ldots, h_d)$ for the support of $H$ and $\bee = (\varepsilon_1, \ldots, \varepsilon_r)$ an ordered basis of the support of the error, we can represent the syndrome $\bs = \be H^\intercal$ as a matrix over the basis $(\bb_d \otimes \bee)$ like we did in (\ref{eq:SyndromeAsMatrix}), i.e.
	$$
	\bs^{\intercal}=
	\begin{pmatrix}
	X_1,  \ldots,  X_d 
	\end{pmatrix}
	(\mathbf{h}_d \otimes \bee)^\intercal.
	$$
	Similarly, the components of the vector $(\mathbf{h}_d \otimes \bee)$ are a generator of the product space $\calH.\calE$ and , when the product space has dimension $dr$, the vector $(\mathbf{h}_d \otimes \bee)$ constitutes an ordered basis of $\calH.\calE.$
	
\end{remark}
\cn

Consider the cases where $\calS \subsetneq \calV_{\alpha,d}.\calE$ and start with the first expansion.
The expansion is given by $\calV_{\alpha,2} \calS = \calS + \alpha \calS \subseteq \calV_{\alpha,d+1}.\calE.$
Note that, since $\calS$ is the support of $\bs$, the support of $\alpha \bs$ is $\alpha \calS.$
As in (\ref{eq:SyndromeSum}), we can express $\alpha \bs$ as
\begin{equation*}\label{eq:AlphaSyndromeSum}
\alpha \bs^\intercal 
=
\alpha X_1 \bee^\intercal + \alpha^2 X_2 \bee^\intercal + \ldots + \alpha^d X_d \bee^\intercal 
\end{equation*}
Denote $\ba_{d+1} = (1,\alpha, \ldots, \alpha^{d})$. Then, if $\dim(\calV_{\alpha,d+1}.\calE) = r(d+1)$, \cn
a basis of $\calV_{\alpha,d+1}.\calE$ can be given by $\ba_{d+1} \otimes \bee = (\ba_d \otimes \bee, \alpha^d \bee) \in \Fqm^{(d+1)r}.$
Thus we have
\begin{equation}\label{eq:AlphaSyndromeAsMatrix}
\alpha \bs^\intercal
=
\begin{pmatrix}
\bzero & X_1 & \ldots & X_d 
\end{pmatrix}
(\ba_{d+1} \otimes \bee)^\intercal.
\end{equation}
Combining (\ref{eq:SyndromeAsMatrix}) and (\ref{eq:AlphaSyndromeAsMatrix}) gives the system
\begin{equation}\label{eq:Exp2AsMatrix}
\begin{pmatrix}
\bs^\intercal \\
\alpha \bs^\intercal
\end{pmatrix}
=
\begin{pmatrix}
X_1 & X_2 & \ldots & X_d & \bzero \\
\bzero & X_1 & X_2 & \ldots & X_d 
\end{pmatrix}
(\ba_{d+1} \otimes \bee)^\intercal.
\end{equation}
We can similarly extend \eqref{eq:Exp2AsMatrix}  for the general expansion $\calV_{\alpha,t}. \calS = \calV_{\alpha,d+t-1}.\calE.$
In this case the coordinates of $\ba_{d+t} \otimes \bee$ generate the space $\calV_{d+t-1}.\calE$ while the coordinates of $(\ba_t \otimes \bs)^\intercal = (\bs, \alpha \bs, \ldots, \alpha^{t-1} \bs)^\intercal$ can be expressed as
\begin{equation} \label{eq:ExpAsMatrix}
\begin{pmatrix}
\bs^\intercal \\ \alpha \bs^\intercal \\ \vdots \\ \alpha^{t-1} \bs^\intercal
\end{pmatrix}
=
\begin{pmatrix}
X_1 & X_2 & \cdots & X_d     & \bzero   & \cdots & \bzero \\
\bzero   & X_1 & X_2 & \cdots & X_d & \cdots & \bzero \\
\vdots   &  & \ddots & \ddots  &  & \ddots & \vdots \\
\bzero   & \cdots &  \bzero & X_1  & X_2 & \cdots & X_d 
\end{pmatrix}
\begin{pmatrix}
\bee^\intercal \\ \alpha \bee^\intercal \\ \vdots \\ \alpha^{d+t-2} \bee^\intercal 
\end{pmatrix}.
\end{equation}

The above expansion is considered successful when
$\calV_{\alpha,t}. \calS = \calV_{\alpha, d+t-1}.\calE$, which
can be stated in terms of generating sets as $$\Span{\Fq}{\ba_t \otimes \bs} = \Span{\Fq}{\ba_{d+t-1} \otimes \bee}.$$
In order to ensure that the above equality holds, we need to consider the big matrix in (\ref{eq:ExpAsMatrix}).
Let $M_t$ be the $(n-k)t \times (d+t-1)r$ matrix on the right hand side of (\ref{eq:ExpAsMatrix}).
If $M_t$ has full rank $(d+t-1)r$, then there will be an invertible square sub-matrix $\hat{M_t}$ of order $(d+t-1)r$ in $M_t$ and a corresponding sub-vector $\hat{\bs} \in \Fqm^{(d+t-1)r}$ such that $\hat{\bs}^\intercal = \hat{M_t} (\ba_{d+t-1} \otimes \bee)^\intercal.$
This implies 
\begin{equation}\label{eq:se}
\hat{M_t}^{-1}\hat{\bs}^\intercal  =  (\ba_{d+t-1} \otimes \bee)^\intercal.
\end{equation}
That is to say, all the entries of $(\ba_{d+t-1} \otimes \bee)$ can be obtained as $\Fq$-linear combinations of entries in $\hat{\bs}.$
In this case, we have that $\Span{\Fq}{\ba_{d+t-1} \otimes \bee} \subseteq \Span{\Fq}{\hat{\bs}}$ 
while $\Span{\Fq}{\hat{\bs}} \subseteq \Span{\Fq}{\ba_t\otimes\bs},$
from which $\Span{\Fq}{\ba_{d+t-1} \otimes \bee} = \Span{\Fq}{\ba_t\otimes\bs}.$ Therefore, 
the successful expansion $\calV_{\alpha,t} \calS = \calV_{\alpha,d+t-1}.\calE$ can be achieved if and only if $M_t$ has full rank $(d+t-1)r$, which is only possible when $(n-k)t\geq (d+t-1)r$.
\color{black}
Observe that one necessary (but not sufficient) condition for $M_t$ to be of full rank $(d+t-1)r$ is $\Rank(X_1) = \Rank(X_d) = r.$ Below we shall investigate the necessary and sufficient 
condition for $\Rank(M_t) = (d+t-1)r$.

\smallskip

Assume $\Rank(X_1) = r$. We can rearrange the rows of $X_1$ such that $X_1 \rightarrow \begin{psmallmatrix}
Y_1 \\ Z_1
\end{psmallmatrix}$ with $\Rank(Y_1)=r$. Furthermore, we can rearrange the rows of $M_t$ using a permutation matrix $P$ to obtain  the matrix
\begin{equation}\label{eq:Mtd_stairs}
P M_t
=
\begin{pmatrix}
Y_1 & Y_2 & \cdots & Y_d       &     &        &     \\
& Y_1 & Y_2    & \cdots    & Y_d &        &     \\
&     & \ddots & \ddots    &     & \ddots &     \\
&     &        & Y_1       & Y_2 & \cdots & Y_d \\
Z_1 & Z_2 & \cdots & Z_d       &     &        &     \\
& Z_1 & Z_2    & \cdots    & Z_d &        &     \\
&     & \ddots & \ddots    &     & \ddots &     \\
&     &        & Z_1       & Z_2 & \cdots & Z_d 
\end{pmatrix},
\end{equation}
where $Y_i \in \Fq^{r\times r}$, $Z_i \in \Fq^{u \times r}$ where $u=n-k-r$ \cn and $\Rank(Y_1)=r$.
Since $Y_1$ is invertible, there exist square matrices  $A_1, \ldots, A_{d-1}$ such that $Y_1 A_{i} = -Y_{i+1}$ for $i=1, \dots, d-1$.
Using the matrices $A_i$ for $1\leq i<d$ and elementary rows and columns operations\cn,
we can transform $PM_t$ to
\begin{equation}\label{eq:Mtd_1}
M_t^{(1)}
=
\begin{pmatrix}
Y_1 &     	\bzero	&    \cdots    &     \bzero      	 &     &        &     \\
& Y_1 		& Y_2    & \cdots    	 & Y_d &        &     \\
&     		& \ddots & \ddots    	 &     & \ddots &     \\
&     		&        & Y_1       	 & Y_2 & \cdots & Y_d \\
& Z_1^{(1)} & \cdots & Z_{d-1}^{(1)} &     &        &     \\
& Z_1 		& Z_2    & \cdots    	 & Z_d &        &     \\
&     		& \ddots & \ddots    	 &     & \ddots &     \\
&     		&        & Z_1       	 & Z_2 & \cdots & Z_d 
\end{pmatrix},
\end{equation}
where $Z_{i}^{(1)} = Z_1 A_i + Z_{i+1}.$
We can continue reducing the second block with the same matrices $A_1, \ldots, A_{d-1}$.
This will give us the further reduction
\begin{equation}\label{eq:Mtd_2}
M_t^{(2)}
=
\begin{pmatrix}
Y_1 &    \bzero &        \cdots		&    \bzero    &     			 &         &     \\
& Y_1 &      \bzero  		&    \cdots    &    \bzero 			 &         &     \\
&     & \ddots      & \ddots   &     & \ddots &     \\
&     &        		& Y_1    & Y_2 			 & \cdots  & Y_d \\
&     & Z_1^{(2)} 	& \cdots & Z_{d-1}^{(2)} &     	   &     \\
&     & Z_1^{(1)}   & \cdots & Z_{d-1}^{(1)} &         &     \\
&     & \ddots 		& \ddots &     			 & \ddots  &     \\
&     &        		& Z_1    & Z_2 			 & \cdots  & Z_d 
\end{pmatrix},
\end{equation}
where $Z_{i}^{(2)} = Z_1^{(1)} A_i + Z_{i+1}^{(1)}$ for $1\leq i<d$ and $Z_{d}^{(1)} = \bzero.$
Let us define $Z^{(j)} \in \Fq^{u \times (d-1)r}$ as $Z^{(j)} = ( Z_1^{(j)} \cdots Z_{d-1}^{(j)})$ and the square matrix $A \in \Fq^{(d-1)r \times (d-1)r}$ as
\begin{equation}\label{eq:matrix_A}
A =
\begin{pmatrix}
A_1 & A_2 & \cdots & A_{d-2} & A_{d-1} 	\\
I_r &     &        &         &  \bzero  \\
& I_r &  	   &         &  \bzero  \\
&	  & \ddots &         &  \vdots  \\
&     &		   & I_r     & 	\bzero	
\end{pmatrix}.
\end{equation}
Observe that $$Z^{(2)} = ( Z_1^{(2)} \cdots Z_{d-1}^{(2)})= ( Z_1^{(1)} \cdots Z_{d-1}^{(1)}) A = Z^{(1)}A.$$ More generally,
we have $Z^{(j+1)} = Z^{(j)} A$, implying that 
$Z^{(j+1)} = Z^{(1)}A^j$ for $j=1, 2,\dots, t-2$. 
Therefore, after iterating the reduction step $t$ times, we obtain the matrix
\begin{equation}\label{eq:Mtd_t}
M_t^{(t)}
=
\begin{pmatrix}
Y_1 &     &        &     &     			  \\
& Y_1 &        &     &     			  \\
&     & \ddots &     &     			  \\
&     &        & Y_1 &  			  \\
&     &        &     & Z^{(1)}A^{t-1} \\
&     &        &     & Z^{(1)}A^{t-2} \\
&     &        &     & \vdots         \\
&     &        &     & Z^{(1)} 	  
\end{pmatrix}.
\end{equation}

Denote \begin{equation}\label{eq:MtZA}
M_t(Z^{(1)}, A) =   \begin{pmatrix}
Z^{(1)} \\ Z^{(1)} A \\ \vdots \\ Z^{(1)} A^{t-1}
\end{pmatrix}
\in \Fq^{ut \times r(d-1)}.
\end{equation}
From \eqref{eq:Mtd_t}
we have
\begin{align}
    \Rank(M_t^{(t)}) & = t \Rank(Y_1) + \Rank(M_t(Z^{(1)},A)) \notag \\ 
    &=  rt + \Rank(M_t(Z^{(1)},A)).
\end{align}
The operations we applied to get $M_t^{(t)}$ from $M_t$ are all linear and invertible, therefore $\Rank(M_t) = \Rank(M_t^{(t)})$.
This leads to the following result.
\begin{proposition}\label{prop:0}
	Let $M_t$ be the matrix in \eqref{eq:ExpAsMatrix} and  $M_t(Z^{(1)},A)$ be the matrix defined in \eqref{eq:MtZA}. If $\Rank(X_1)=r$ and  $\Rank(M_t(Z^{(1)},A)) = (d-1)r$ then the equality $\calV_{\alpha, t}.\calS = \calV_{\alpha, d+t-1}.\calE$ holds.
\end{proposition} 

Note that it is still possible that $\calV_{\alpha, t}.\calS = \calV_{\alpha, d+t-1}.\calE$ with $\dim(\calV_{\alpha, d+t-1}.\calE) < (d+t-1)r.$
In this case, in order to have a successful expansion, it will be sufficient that $\Rank(M_t) = \dim(\calV_{\alpha, d+t-1}.\calE).$
For the rest of this paper we will consider the worst case where $\dim(\calV_{\alpha, d+t-1}.\calE) = (d+t-1)r.$
This means we need to estimate the probability that $\Rank(M_t) = (d+t-1)r$ which, thanks to
Proposition \ref{prop:0}, can be transformed to establish the probability that $\Rank(X_1) = r$ and that $\Rank(M_t(Z^{(1)},A)) = (d-1)r.$
In our analysis, we will consider the matrices $X_1, \ldots, X_d$ as independently chosen from a uniform distribution.
This assumption is true for certain parameters as in the following lemma.
\cn

%
In \cite[Th.\,11]{renner20} and later in \cite[Prop.\,3]{Burle23} the authors established the probability that the syndrome support is equal to the product space $\Span{\F_q}{H}.\calE$.
Concretely, they proved that the entries of the syndrome can be considered independently and uniformly chosen from the product space.   
For the sake of completeness, we report here in our notation the result in \cite{Burle23}, which shows that $X_1, \ldots, X_d$ can be deemed as matrices independently and uniformly chosen from $\Fq^{(n-k)\times r}$. 
\begin{lemma}\label{lm:SyndromeUniform}
	Consider two spaces $\calE$ and $\calV_{\alpha,d}$ having basis $\bm \varepsilon = (\varepsilon_1, \ldots, \varepsilon_r)$ and $\bm a_d = (1, \ldots, \alpha^{d-1})$ respectively.
	Let $\be = \bm \varepsilon E$ for $E\in \F_q^{r\times n}$ an error of rank $r$ and let $H \in \Fqm^{(n-k) \times n}$ be a matrix  such that its rows $\bm h_{1}, \ldots \bm h_{n-k}$ can be written as $\bm h_i = \bm a_{d} M_i$ for $M_i\in \F_q^{d\times n}$, \,$1\leq i\leq n-k$.
	If the entries of $E$ and $M_i$ are independently and uniformly distributed random variables over $\Fq$, then the entries of the vector 
	$$
	\bs = \bm \varepsilon E \begin{pmatrix}M_1^\intercal & \ldots & M_{n-k}^\intercal \end{pmatrix} \begin{pmatrix}
	\bm a_d^\intercal &  & \\
	& \ddots & \\
	& & \bm a_d^\intercal,
	\end{pmatrix}
	$$
	are independently and uniformly distributed random variables over $\calV_{\alpha,d}.\calE.$
\end{lemma}
\begin{proof}
	Refer to the proof of \cite[Proposition.3]{Burle23}.
\end{proof}
As a direct consequence, when both the parity-check matrix and the error are chosen independently and uniformly at random from $\calV_{\alpha,d}^{(n-k)\times n}$ and $\calE^n$, respectively, the syndrome of a BD-LRPC code can be described as
$$
\bs^\intercal = \begin{pmatrix}
X_1 & \ldots & X_{n-k}
\end{pmatrix} (\bm a_d \otimes \bm \varepsilon)^\intercal,
$$
where $X_i$ are independently and uniformly chosen from $\Fq^{(n-k) \times r}.$
\cn

We are left to investigate the probability of $\Rank(M_t(Z^{(1)},A)) = (d-1)r$ as characterized in Proposition \ref{prop:0}.
We will divide our discussion in two cases. For the case of $d=2$, the matrices $Z^{(1)}, A$ can be deemed 
as uniformly distributed, and the probability is analyzed in great detail in Section \ref{sec:ExpRowSub}; and for the cases that $d>2$,  we investigate the probability by both theoretic analysis and computational simulations.

\subsection{The case for $d=2$}\label{subsec:d=2}
It is worth noting that, in the special case $d=2$, the matrix $\begin{psmallmatrix}
A \\ Z^{(1)} 
\end{psmallmatrix},$ can be considered as uniformly drawn in $\Fq^{(r+u) \times r}.$  
Indeed, for $d=2$, we have that $Z^{(1)} = Z_1 A_1 + Z_2$ while $Z^{(i)} = Z^{(1)} A_1^{i-1}$ and $A = A_1.$
Assuming that the matrices $X_1,X_2$ derived from $\bs=\bs_1+\alpha \bs_2$ follow a uniform distribution in $\Fq^{(r+u) \times r}$, the matrices $Z^{(1)}=Z_1A_1+Z_2 \in\Fq^{u \times r}$ and $A_1 \in \Fq^{r \times r}$  can be deemed to follow a uniform distribution. 
As a matter of fact, assume $X_1$ satisfies the necessary condition $\Rank(X_1) = r.$
By $\begin{psmallmatrix}
Y_2 \\ Z_2
\end{psmallmatrix}
=P_0X_2
$
and
$A = -Y_1^{-1} Y_2$, we have
\begin{align}\label{eq:AZ_uniform}
    \begin{pmatrix}
    A \\ Z^{(1)}
    \end{pmatrix} 
     & =
    \begin{pmatrix}
    -Y_1^{-1} & \bzero \\
    - Z_1 Y_1^{-1} & I_{u}
    \end{pmatrix}
    \begin{pmatrix}
    Y_2 \\ Z_2
    \end{pmatrix} \notag \\
    & = 
    \begin{pmatrix}
    -Y_1^{-1} & \bzero \\
    - Z_1 Y_1^{-1} & I_{u}
    \end{pmatrix}
    P_0X_2.
\end{align} 
It is readily seen that 
$\begin{psmallmatrix}  A \\ Z^{(1)} \end{psmallmatrix}$ 
has a one-to-one correspondence to the uniformly distributed $X_2$. Therefore,  $\begin{psmallmatrix}  A \\ Z^{(1)} \end{psmallmatrix}$  also follows a uniform distribution.

A complete description of the probability that the rows of $M_t(Z^{(1)},A)$ span the whole space $\Fq^r$ in the case $d=2$ will be done in Section \ref{sec:ExpRowSub}.
Since the case $d=2$ covers all the LRPC codes of density $2$ we consider this case as the most important.
For $d>2$ the family of BD-LRPC codes represent a small fraction of all the possible LRPC codes of density $d$, especially for large value of $d,$ moreover the minimal value of the parameter $m$ becomes quickly impractical as $d$ grows.

\subsection{The case for $d>2, r=1$}\label{subsec:d>2}
In this case, the matrix $A$ in equation (\ref{eq:matrix_A}) becomes increasingly sparse as $d$ grows.
Consequently, the tools developed in Section \ref{sec:ExpRowSub} are not suitable to directly analyze this case. Nevertheless, we can provide a rigorous description for the sparsest case.

For simplicity of notation, within this section we denote $n = d-1.$
We can consider the case $r=1$ for which $A'$ is in the form
$$
A' = \begin{pmatrix}
a_{n-1} & a_{n-2} & \cdots & a_{1} & a_{0}\\
1   & & & & 0\\
& 1 &&& 0 \\
&& \ddots && \vdots \\
&&& 1&0 \\
\end{pmatrix}.
$$
As we will see in Lemma \ref{lm:BasisChangeInv}, for any $Z  \in \Fq^{u \times n}$ we have that the dimension of the space generated by the rows of $Z, ZA, \ldots, ZA^{n-1}$ is the same as the dimension of the space generated by the rows of $ZB(B^{-1} A' B)^i$ whenever $B$ is invertible.
If $Z$ is drawn from a uniform distribution so it is for $ZB.$
Take $B$ the invertible matrix having $1$ on its anti-diagonal and zero elsewhere. Then the matrix 
\begin{equation}\label{eq:comp_mat}
A  = B^{-1}A'B = \begin{pmatrix}

0& 1   & & \\
0&& 1 && \\
0&&& \ddots & \\
0&&& &1 \\
a_{0} & a_{1} & \cdots & a_{n-2} & a_{n-1}\\
\end{pmatrix}
\end{equation}
is the companion matrix for the polynomial $p_A(x) = x^{n} - a_{n-1} x^{n-1} - \cdots - a_0.$

Consider the vector $\bz = (z_0, \ldots, z_{n-1}).$ We can interpret $\bz$ as the coefficients of the polynomial $z(x) = z_{n-1} x^{n-1} + \ldots + z_0.$
Observe that the polynomial $z(x) x  \mod p_A(x)$ will have coefficient representation $\bz A$. In general $z(x) x^i$ will have coefficient representation given by $\bz A^i.$
The vector space $\Omega(\bz, A) \subseteq \Fq^n$ generated by $\bz, \bz A, \ldots, \bz A^{n-1}$ corresponds to the ideal generated by $z(x)$ in $\Fq[x]/p_A(x).$
This means that $\Omega(\bz,A)= \Fq^n$ if and only if $\gcd(z(x), p_A(x)) = 1.$
The probability of two monic polynomials in $\Fq[x]$ not both of degree zero to have no common divisor is $1- \frac{1}{q}$ \cite{Bennett07}.
From the proof in \cite{Bennett07} it is easy to see that this is also true if we consider the probability where one polynomial is monic of degree $n \geq 1$ and the other is a random polynomial (not necessarily monic) of degree less than $n$.

Consider $Z \in \Fq^{u \times n}$, and let $\Omega(Z,A) \subseteq \Fq^n$ be the space generated by the rows of $Z, ZA, \ldots, ZA^{n-1}.$
This space will correspond to the ideal of $\Fq[x]/p_A(x)$ generated by $z_1(x), \ldots, z_u(x)$ where $z_i(x)$ is the polynomial represented by the $i$-th row of $Z.$
The probability that $\Omega(Z,A) = \Fq^n$ correspond to the probability that $\gcd(p_A(x), z_1(x), \ldots, z_u(x)) = 1.$
The probability that $u+1$ monic polynomials of $\Fq[x]$ not all of degree zero have $\gcd$ equal to $1$ is estimated as $1-\frac{1}{q^u}$ \cite{Bennett07}.
Even in this case it is easy to see that the same result holds when we consider polynomials that are not necessarily monic or, as in our case, one polynomial is monic of degree $n \geq 1$ and the other polynomials are not necessarily monic of degree less than $n$.

For the uniformly distributed matrices $A$, Corollary \ref{cor:Pr(C_k^u)}
in Section \ref{sec:ExpRowSub} shows that the probability of $\Omega(Z,A) = \Fq^n$ is approximately $1-\frac{1}{(q-1)q^{u-1}}$.
It turns out that in the above extreme case, we obtain slightly better probability $1-\frac{1}{q^u}$.


\medskip

Besides the theoretical analysis above, computational simulation seems to confirm that intermediate cases follow a similar probability as the other two cases. However, providing a formal proof for those intermediate cases seems intractable. We consider it is an interesting open problem.

\cn

\section{Probability of Successful Expansions}\label{sec:ExpRowSub}

This section is dedicated to studying the
following problem, which is derived from Proposition \ref{prop:0} in Section \ref{sec:Successful} for successful expansions.
\begin{problem}\label{Pb:C_k}
   Let $u, r, k, t$ be positive integers with $k\leq r$.
   Determine the number of matrix pairs $(Z,A) \in \Fq^{u \times r} \times \Fq^{r \times r}$ such that the following matrix 
   \begin{equation}\label{eq:matrix_expansion}
   M_t(Z,A) = \begin{pmatrix}
       Z \\ ZA \\ \vdots \\ Z A^{t-1} 
   \end{pmatrix}
 \end{equation} has $\Rank(M_t(Z,A))=k$.
   Equivalently, define \begin{equation}\label{eq:Omega}
        \Omega_t(Z,A) = \RowSpan(Z)+ \dots + \RowSpan(Z A^{t-1}),
    \end{equation}  and
   \begin{equation}\label{eq:Ckurt}
       C_k^{(u,r,t)} = \left\{(Z,A) \in \Fq^{u \times r} \times \Fq^{r \times r} \mid \dim(\Omega_t(Z,A)) = k \right\}
   \end{equation}
   Determine the cardinality of the set $ C_k^{(u,r,t)}.$
\end{problem}
Note that the positive integer $k$ here  is independent of the code dimension in other sections. 
To the best of our knowledge, this counting problem was not present in literature except for the case $u = 1$ and $A$ being an invertible matrix.
From a purely mathematical perspective, we consider this counting problem to be interesting on its own. 
In the context of decoding BD-LRPC codes, we will investigate the value of $|C_k^{(u,r,t)}|$ for the case $t \geq r$ (as in Theorem \ref{th:|C_k^ur|})  and we will give a lower bound for the case $t=\lceil \frac{r}{u}\rceil + 1$ (as in Theorem \ref{th:|Gammak^t|}). 
The investigation is technical and we provide some auxiliary results first.
\subsection{Auxiliary Results}
Some properties of the subspace $\Omega_t(Z,A)$ are discussed below. 
\begin{lemma}\label{lm:calV_r-1}
    Let $\Omega_{t}(Z,A)$ be defined as in \eqref{eq:Omega}. Then, 
    \begin{enumerate}[label=(\roman*)]
        \item  it can be equivalently defined in a recursive way: $\Omega_{1}(Z,A) = \RowSpan(Z)$ 
        and \begin{equation}\label{eq:V_t+1}
        \Omega_{t+1}(Z,A) = \Omega_{t}(Z,A) + \Omega_{t}(Z,A)A  \text{ for any } t\geq 1,
    \end{equation} where the notation $\Omega_{t}(Z,A)A$ denotes the image of $\Omega_{t}(Z,A)$ under the linear map $R_A:\Fq^r \rightarrow \Fq^r$ given by $R_A(\bv) = \bv A$ for any $\bv \in \Fq^r$;
    \item $\Omega_{1}(Z, A) \subseteq \cdots \subseteq \Omega_t(Z, A)$ for any integer $t\geq 1$; moreover, let $d_i = \dim(\Omega_{i+1}(Z, A))-\dim(\Omega_i(Z,A))$ for $1\leq i<t$, then $ u \geq d_{1}\geq  \cdots \geq d_{t-1}$;
    \item if $\Omega_{t}(Z,A)=\Omega_{t+1}(Z,A)$ for certain integer $t$,  then $\Omega_{t}(Z,A) = \Omega_{t'}(Z,A)$ for any $t' \geq t+1.$
    Such an integer $t$ always exists and satisfies $t<r$.
    \end{enumerate}
\end{lemma}

\begin{proof}
(i) Given $\Omega_{t}(Z,A) = \Span{\Fq}{\bv_1, \ldots, \bv_{k_t}}$,  its image under the linear map $R_A$ is given by
$\Omega_{t}(Z,A)A = \Span{\Fq}{\bv_1 A, \ldots, \bv_{k_t} A}.$
By \eqref{eq:Omega} we see that $\Omega_{t}(Z,A) = \sum_{i=0}^{t-1}\RowSpan(ZA^i)$ which is equal to $\Span{\Fq}{\bz_j A^i \mid j \in [u], i \in  \{ 0, \ldots, t-1 \} }$, where $\bz_j$ corresponds to the $j$-th row of $Z$.
Then the subspace $\Omega_{t}(Z,A) + \Omega_{t}(Z,A)A$ is equal to $\Span{\Fq}{\bz_j A^i \mid j \in [u], i \in \{ 0, \ldots, t-1 \}} + \Span{\Fq}{\bz_j A^{i+1} \mid j \in [u], i \in \{ 0, \ldots, t-1 \}}$ which in turn is equal to 
$\Span{\Fq}{\bz_j A^i \mid j \in [u], i \in \{ 0, \ldots, t \}} = \Omega_{t+1}(Z,A).$    

(ii) The relation $\Omega_1(Z,A) \subseteq \cdots \subseteq \Omega_t(Z,A)$ follows directly from (i). Furthermore, 
for $i\geq 2$, there exists $\calV_i = \Span{\Fq}{\bv_1, \ldots, \bv_{d_i}}$ such that $\Omega_{i}(Z,A)$ can be decomposed in the direct sum $\Omega_{i}(Z,A) = \Omega_{i-1}(Z,A) + \calV_i.$
   Using the fact $\Omega_{i-1}(Z,A)A \subseteq \Omega_{i}(Z,A),$ we obtain
   \begin{align*}
       \Omega_{i + 1}(Z,A)  & = \Omega_{i}(Z,A) + \Omega_{i}(Z,A)A \\
                           & = \Omega_{i}(Z,A) + (\Omega_{i-1}(Z,A) + \calV_i) A \\
                           & = \Omega_{i}(Z,A) + \calV_i A.
    \end{align*}
    This implies that 
    $
        d_i = \dim(\calV_i) \geq \dim(\calV_iA ) \geq 
        d_{i+1}.
    $
    The last inequality comes from the fact that $\dim(\calV_{i+1}) = \dim( \calV_iA) - \dim(\Omega_{i}(Z,A) \cap \calV_i A).$

(iii)
    When $\Omega_{t}(Z,A) = \Omega_{t+1}(Z,A)$ holds, 
    from the recursive relation in (\ref{eq:V_t+1}) we have
    $$
        \Omega_{t+2}(Z,A) = \Omega_{t}(Z,A) + \Omega_{t}(Z,A)A = \Omega_{t+1}(Z,A).   
    $$
    This means that the chain stops to grow at $\Omega_{t}(Z,A)$.
    Thus we have  $\Omega_{t'}(Z,A) = \Omega_{t}(Z, A)$ for all $t'>t.$
    Finally,
    if $\dim(\Omega_1(Z,A)) \geq 1,$ then $\dim(\Omega_{t'}(Z,A)) \geq t'$ or there exist $1 \leq t < t'$ such that $\Omega_{t'}(Z,A) = \Omega_{t}(Z,A)$ and $\dim(\Omega_{t}(Z,A)) \geq t.$
    Since $\Omega_{t}(Z,A) \subseteq \Fq^r,$ its dimension is upper bounded by $r.$
    Then from $r \geq \dim(\Omega_{t}(Z,A)) \geq t$ we obtain that $t \leq r.$
    For the particular case of $\Omega_1(Z,A) = \{ \bzero \},$ we have $Z=\bzero$ and $ \Omega_{t'}(Z,A) = \{ \bzero \}$ for all $t' \geq 1.$
\end{proof}
As a direct consequence of (iii), we see that the expansion of $\Omega_i(Z,A)$ stops at certain integer $ t\leq r$. That is to say, $\Omega_t(Z,A)=\Omega_{r}(Z,A)$ for any $t\geq r$ and it suffices to consider $\Omega_{r}(Z,A)$ as the maximal expansion.
It is clear that $\Omega_{r}(Z,A) $ is the smallest subspace $\Omega \subseteq \Fq^r$ such that $\RowSpan(Z) \subseteq \Omega$ and $\Omega A = \{\bx A | \bx \in \Omega\}  \subseteq \Omega.$ For simplicity, we denote $\Omega_{r}(Z,A)$  as $\Omega(Z,A)$ in the sequel and 
denote $C_k^{(u,r)} = C_k^{(u,r,r)}$ where
\begin{equation}\label{eq:Ckur}
C_k^{(u,r)} = \{(Z,A) \in \Fq^{u \times r} \times \Fq^{r \times r} \mid \dim(\Omega(Z,A)) = k \}.
\end{equation}

The following property of the set $C_k^{(u,r,t)}$ will facilitate our subsequent calculations.

\begin{lemma}\label{lm:BasisChangeInv}
    Let $C_k^{(u,r,t)}$ be defined as in \eqref{eq:Ckurt}
    and $B$ be an invertible matrix in $\Fq^{r\times r}.$
    For any  $(Z,A) \in C_k^{(u,r,t)}$, we have $(ZB, B^{-1} A B) \in C_k^{(u,r,t)}$. 
\end{lemma}
\begin{proof}
    For $(Z,A) \in C_k^{(u,r,t)}$, we can express $\Omega_t(Z,A)$ as the row span of the matrix $M_t(Z,A)$ \begin{equation}\label{eq:calV_Matrix}
        \Omega_t(Z,A) = \RowSpan 
        \begin{pmatrix}
        Z \\
        Z A \\
        \vdots \\
        Z A^{t-1}
        \end{pmatrix}.
    \end{equation}
    Notice that $(B^{-1} A B) (B^{-1} A B) = B^{-1} A^2 B$, implying $(B^{-1} A B)^i = B^{-1} A^i B.$ This gives
    $ZB (B^{-1} A B)^i = Z A^i B.$  
    Hence, for the pair $(ZB, B^{-1} A B)$ we obtain
    $$
        \Omega_t(ZB, B^{-1} A B) = \RowSpan 
        \begin{pmatrix}
        Z B \\
        Z A B \\
        \vdots \\
        Z A^{t-1} B
        \end{pmatrix}
        =
        \Omega_t(Z, A) B.
    $$
    Since $B$ is invertible, it holds $\dim( \Omega_t(Z, A)B) =  \Omega_t(Z, A) = k,$ implying that $(ZB, B^{-1} A B) \in C_k^{(u,r,t)}.$
\end{proof}

\begin{algorithm}[ht]
\small
\SetAlgoLined
\KwIn{A matrix $Z \in \Fq^{u \times r}$, a matrix $A \in \Fq^{r \times r}$ and an integer $t.$}
\KwOut{A matrix $G$ such that $\RowSpan(G) = \Omega_t(Z,A)$.}
\tcp{Initialize $\Omega$ as the zero subspace and $G$ as an empty list.}
$\Omega = \{ \bzero \}$ \;
$G = []$ \;
\For{$j \in [0 \ldots  t-1]$}
{   
	\tcp{Add all the new linearly independent rows obtained from $ZA^j$ using the convention $A^0 = I_r.$}
	$G^{(j)} = []$\;
	\For{$i \in [u]$}
	{
		\If{$\bz_i A^j \notin \Omega$}
		{
			$G^{(j)}.append(\bz_i A^j)$\;
			$\Omega = \Omega + \Span{\Fq}{\bz_iA^j}$\;
		}
	}
	$G.append(G^{(j)})$\;
}
Return $G$\;
\caption{Generation of a basis of $\Omega_t(Z,A)$ for  $(Z, A),t$}\label{alg:SeqGen} 
\normalsize
\end{algorithm}

To solve Problem 1, we need to consider linearly independent rows in the matrix $M_t(Z,A)$ as it expands.
This is a process of extracting those linearly independent row vectors in $ZA^j$ for $j=0,\dots, t-1$ in sequential order for a matrix pair $(Z, A)$.
This process is summarized in Algorithm \ref{alg:SeqGen}.
The following lemma shows that the matrix $G,$ obtained as output from Algorithm \ref{alg:SeqGen}, generates the subspace $\Omega_t(Z,A)$ of dimension $k$ on input $(Z,A), t$ when $(Z,A)$ is taken from $C_k^{(u,r,t)}$.
\begin{lemma}\label{lm:normalized_basis}
	For $(Z,A)$ in $C_k^{(u,r,t)}$,
	the rows of the output matrix $$G=\begin{pmatrix}
	G^{(0)} \\ \hline  G^{(1)} \\ \hline   \vdots \\ \hline G^{(t-1)} 
	\end{pmatrix}$$ of Algorithm \ref{alg:SeqGen} form a basis of $\Omega_t(Z,A) \subseteq \Fq^r$. In particular, 
	when $t=r$, the rows of $G$ form a basis of $\Omega(Z,A)$.
\end{lemma}
\begin{proof}
	
	
	
	In Algorithm \ref{alg:SeqGen}, the vector space $\Omega$ is initialized as $\Omega = \{ \bzero \}.$
	Using the convention $A^0 = I_r$, in the initial execution of the inner loop, we add to $G^{(0)}$ all the linearly independent rows of $ZA^{0} = Z.$
    Each time a new row is added to $G^{(0)}$ the subspace $\Omega$ is enlarged to encompass that row.
    By the end of the first inner loop we will have $\Omega=\RowSpan(G^{(0)}) = \RowSpan(Z) = \Omega_0(Z,A).$
	On the second iteration $j=1$, all the rows of $ZA$ that do not belong to $\Omega$ will be added to $G^{(1)}$.
    For each row we add to $G^{(1)}$, the subspace $\Omega$ is extended to encompass that row as well.
	At the end of this iteration, the subspace $\Omega$ correspond to $\RowSpan(G)$ for $G=\gbinom{G^{(0)}}{G^{(1)}}$ which is equal to $\Omega_{2}(Z,A).$
	Continuing in this way, after the execution of the $j$-th inner loop we have $\Omega =\RowSpan(G) = \Omega_j(Z,A)$ for $G = \begin{psmallmatrix}
	    G^{(0)} \\ \vdots \\ G^{(j-1)}
	\end{psmallmatrix}.$
	The statement thus follows.
\end{proof}

Below we provide an example to show the steps in Algorithm \ref{alg:SeqGen}.

\begin{example}\label{Ex1}
	 Let $u=3$, $r=4$,  $t=3$ and $k=3$. Assume 
	$$
	Z = \begin{pmatrix}
			1 & 0 & 0 & 0 \\
			0 & 1 & 0 & 0 \\
			1 & 1 & 0 & 0
	\end{pmatrix},\quad
	A = \begin{pmatrix}
	0 & 0 & 1 & 0 \\
	1 & 1 & 1 & 0 \\
	0 & 0 & 0 & 1 \\
	1 & 1 & 1 & 1
	\end{pmatrix}.
	$$
Then, from Algorithm \ref{alg:SeqGen} one obtains
$$
G = \begin{pmatrix}
G^{(0)} \\ \hline  G^{(1)} \\ \hline G^{(2)} \cn
\end{pmatrix} 
=
\begin{pmatrix}
1 & 0 & 0 & 0 \\
0 & 1 & 0 & 0 \\ \hline 
0 & 0 & 1 & 0 \\ \hline 
0 & 0 & 0 & 1 
\end{pmatrix}.
$$ Note that in Algorithm \ref{alg:SeqGen} $G^{(3)}$ in iteration $j=3$ is an empty matrix.
\end{example}

We now take a closer look at the steps of Algorithm \ref{alg:SeqGen}. Denote $Z^{(0)} = Z$ and let $G^{(0)}$ be the matrix obtained after the iteration $j=0$.
Notice that $G^{(0)}$ are sequentially formed by  the rows in $Z^{(0)}$ linearly independent from the rows in $Z^{(0)}$ before them; and that any row of $Z^{(0)}$ is a linear combination of rows in $G^{(0)}.$ 
For instance, as in Example \ref{Ex1} we have
$$
Z^{(0)}=\begin{pmatrix}
1   & 0   & 0   & 0   \\
0   & 1   & 0   & 0   \\
1   & 1   & 0   & 0   \\
\end{pmatrix}
\Longrightarrow G^{(0)} = 
\begin{pmatrix}
1   & 0   & 0   & 0   \\
0   & 1   & 0   & 0   \\ 
\end{pmatrix}
$$
It is clear that, all the rows in matrix $Z^{(0)}A$ can be linearly expressed by rows in $Z^{(1)}=G^{(0)}A$. Similarly, denote by $G^{(1)}$ the matrix sequentially formed by rows in $Z^{(1)}$, which cannot be linearly expressed by rows in $G^{(0)}$ and preceding rows in $Z^{(1)}$.
For instance, 
$$
Z^{(1)} = G^{(0)}A = 
\begin{pmatrix}
0 & 0 & 1 & 0 \\
1 & 1 & 1 & 0
\end{pmatrix}
\Longrightarrow G^{(1)} = 
\begin{pmatrix}
0 & 0 & 1 & 0
\end{pmatrix}.
$$
We continue this process as in Algorithm \ref{alg:SeqGen} and 
denote $Z^{(j)}=G^{(j-1)}A$ for $1\leq j<r-1$ and $G^{(j)}$
the matrix that sequentially extracts rows from $Z^{(j)}$ which cannot be linearly expressed by rows in $G^{(0)}, \dots, G^{(j-1)}$ and preceding rows in $Z^{(j)}$. Then from Algorithm \ref{alg:SeqGen} we have
\begin{equation}\label{eq: matrix_relation}
\begin{pmatrix}
Z^{(0)} \\ \hline 
Z^{(1)} \\ \hline \vdots  \\
\hline
Z^{(r-1)} 
\end{pmatrix}
\Longrightarrow
G = 
\begin{pmatrix}
G^{(0)} \\ \hline 
G^{(1)} \\ \hline \vdots  \\
\hline
G^{(r-1)} 
\end{pmatrix},
\end{equation} where for $j$ larger than certain value $t$, the matrices $G^{(j)}$ could be empty matrices.

From the above process, it is easily seen that different pairs of $(Z,A)$ can produce the same output $G \in \Fq^{k \times r}$ by Algorithm \ref{alg:SeqGen}.
The number of pairs producing the same output $G \in \Fq^{k \times r}$ turns out to be \color{black} the same for all the possible matrices $G' \in \Fq^{k \times r}$ of rank $k$. \color{black}
This fact will be shown in the following proposition,
which enables us to concentrate on the special output matrix $E_k \in \Fq^{k \times r}$, formed by the first $k$ rows of the identity matrix $I_r \in \Fq^{r \times r}$,
in our calculations.

\begin{lemma}
	\label{prop:BasisInvAlg}
	Let $G\in \Fq^{k \times r}$ be a matrix of rank $k \leq r$ and
	denote 
	$$
	C_G^{(u,r,t)} = \{(Z,A) \in \Fq^{u \times r} \times \Fq^{r \times r}: Alg(Z,A,t) = G\}.
	$$
	Let 
	$E_k$ be the submatrix formed by the first $k$ rows of $I_r$ over $\Fq$. Then we have $$|C_G^{(u,r,t)}| = |C_{E_k}^{(u,r,t)}|.$$
\end{lemma}
\begin{proof}
	To prove this result we will create a bijection between $C_G^{(u,r,t)}$ and $C_{E_k}^{(u,r,t)}.$
	Since $k \leq r$, there exists an invertible matrix $B \in \Fq^{r \times r}$ such that $GB=E_k.$ 
	Consider the map $B: (Z,A) \rightarrow (ZB, B^{-1} A B),$ since $B$ is invertible this is a bijection of $\Fq^{u \times r} \times \Fq^{r \times r}$ to itself.
	We will prove that, for each $(Z,A)$ such that $Alg(Z,A,t) = G,$ we have that $Alg(ZB,B^{-1}AB,t) = GB = E_k.$
	
	We compare the procedure of Algorithm \ref{alg:SeqGen} on two inputs $(Z,A,t)$ and $(ZB,B^{-1}AB,t),$ showing that, when the first yields $G$ as output, the second will yield $E_k.$
	Let $\bz_i$ be a row of $Z$ that is linearly independent from $\bz_1, \ldots, \bz_{i-1},$ then $\bz_i B$ is also independent from $\bz_1 B, \ldots, \bz_{i-1} B.$
	This means that, every time we add a row $\bz_i$ in the first execution we add $\bz_i B$ in the second as well.
    In this way, at each step, $G'B = E_{k'},$ where $G'$ and $E_{k'}$ are the matrices obtained so far in the algorithm.
	At the end of the first inner loop, if $\Rank(Z) = k_0 \leq u,$ we will have $G^{(0)}B = E_{k_0}$ where both $G^{(0)},E_{k_0} \in \Fq^{k_0 \times r}.$
	
	In the $j$-th inner loop, the algorithm will do the following.
	On initial input $(Z,A)$, it will compute $ZA^j$ then, for each row vector $\bz_i A^j$, it will check if it is already in $\Omega = \Span{\Fq}{\bg_1, \ldots, \bg_{k'}}$, if not, it will add it as the last row of $G'$ and update $\Omega B$ to $\Span{\Fq}{\bg_1, \ldots, \bg_{k' + 1}}$.
	On initial input $(ZB,B^{-1}AB)$, it will compute $ZB(B^{-1}AB)^j$ then, for each row vector $\bz_iB (B^{-1}A B)^j$, it will check if it is already in $\Omega B = \Span{\Fq}{\be_1, \ldots, \be_{k'}}$, if not, it will add it as the last row of $E_{k'}$ and update $\Omega B$ to $\Span{\Fq}{\be_1, \ldots, \be_{k'}, \bz_iA^j B}$ 
	In the proof of Lemma \ref{lm:BasisChangeInv} we have already seen that $\bz_i B (B^{-1}A B)^j = \bz_i A^j B.$
    \color{black}Notice that, if $\bz_i A^j = g_{k'+1}$ then $\bz_i A^j B = \be_{k'+1}$, then $E_{k'}$ is updated to $E_{k' + 1}$.\color{black}

	We will show that, at each step, $G'B = E_{k'}.$
	Let $G_{j-1}$ and $E_{j-1}= G_{j-1}B$ be the matrices obtained after we complete the inner loop for the $(j-1)$-th time.
	It is easy to see that 
	$$
	\bz_i A^j \in \RowSpan(G_{j-1}) + \Span{\Fq}{\bz_1, \ldots, \bz_{i-1}}A^j,
	$$
	if and only if
	$$
	\bz_i A^j B \in \RowSpan(E_{j-1}) + \Span{\Fq}{\bz_1, \ldots, \bz_{i-1}}A^j B.
	$$
	where we used the fact that $\RowSpan(E_{j-1}) = \RowSpan(G_{j-1})B.$
\end{proof}
Based on Lemma \ref{prop:BasisInvAlg} we have the following corollary, which significantly simplifies our calculations.
\begin{proposition}
	\label{cor:AqGammak}
	Let $C_k^{(u,r,t)}$
	be given as in \eqref{eq:Ckurt}
	and 
	$$
	   C_{E_k}^{(u,r,t)} = \{(Z,A) \in \Fq^{u \times r} \times \Fq^{r \times r}: Alg(Z,A,t) = E_k\}.
	$$
	Then we have
	$$
	   |C_k^{(u,r,t)}| = A_q(r,k) |C_{E_k}^{(u,r,t)}|.
	$$
\end{proposition}
\begin{proof}
	A pair $(Z,A) \in C_k^{(u,r,t)}$ iff its output $Alg(Z,A,t)$ through Algorithm \ref{alg:SeqGen} is a full rank matrix $G \in \Fq^{k \times r}$.
	There are $A_q(k,r)$ distinct matrices in $\Fq^{k \times r}$ of rank $k$. 
	For each distinct matrix $G \in \Fq^{k \times r}$ of rank $k$, Lemma \ref{prop:BasisInvAlg} implies that 
	$|C_{G}^{(u,r,t)}|=|C_{E_k}^{(u,r,t)}|$.
	Hence there are $|C_{E_k}^{(u,r,t)}|$ distinct tuples $(A,Z)$ that generate each full-rank matrix in $\Fq^{k\times r}$. The desired statement thus follows.
\end{proof}

By Proposition \ref{cor:AqGammak}, to solve Problem \ref{Pb:C_k}, it suffices to find the value of $|C_{E_k}^{(u,r,t)} |.$
We shall investigate $|C_{E_k}^{(u,r,t)} |$ for some integers $t$ that are particularly interesting in the decoding of BD-LRPC codes.


\subsection{The cardinality of $C_{k}^{(u,r, t)}$ for $\color{black}t\geq r$}
This subsection presents an explicit formula to compute the cardinality of $C_k^{(u,r)}$ defined in (\ref{eq:Ckur}) and the detailed calculations.

\begin{theorem}\label{th:|C_k^ur|}
	The set $C_k^{(u,r)}$ defined in \eqref{eq:Ckur} has cardinality 
	\begin{align*}
	    &\left|C_k^{(u,r)}\right| =  A_q(r,k)\begin{bmatrix} k + u - 1 \\ u - 1\end{bmatrix}_q q^{r(r-k) + k} \\
        & =\begin{bmatrix} k + u - 1 \\ u - 1\end{bmatrix}_q \frac{H_q(r)}{H_q(r-k)}  q^{rk} q^{r(r-k) + k}.
	\end{align*}
\end{theorem}
The calculation of the formula in Theorem \ref{th:|C_k^ur|} involves several parameters. 
Below we shall first discuss the calculation of $|C_{E_k}^{(1,r)}|$, which is easier to follow but exhibits the key idea for the general case $u\geq 2$.
Then we will present the calculation of $C_{E_k}^{(u,r)}$ with the help of Ferrers diagrams.

\subsubsection{The case $u=1.$}

For $u=1,$ the matrix $Z \in \Fq^{u \times r}$ is just a row vector $\bz \in \Fq^{ r}.$
Let $\by_i = \bz A^{i-1}$ and consider the matrix $Y \in \Fq^{r \times r}$ whose rows are the vectors $\by_i$.
Following Proposition \ref{cor:AqGammak}, we focus on the case $G = E_k$, where $G$ is the output of Algorithm \ref{alg:SeqGen} on input $(\bz,A),r$.
We can characterize the structure of $\bz$ and $A$ that gives this output.
We have
$$
Y = \begin{pmatrix}
\by_1 \\ \by_2 \\ \vdots \\ \by_{r}
\end{pmatrix} = 
\begin{pmatrix}
\bz \\ \bz A \\ \vdots \\ \bz A^{r-1}
\end{pmatrix} \Longrightarrow
G = E_k = \begin{pmatrix}
\be_1 \\ \be_2 \\ \vdots \\ \be_k
\end{pmatrix}.
$$
At first, it is clear that $\bz = \be_1$, otherwise one cannot obtain $G=E_k$, implying $\bz=\by_{1} =\be_1 $. 
Then the second row $\by_2= \be_1A$ is the first row of $A$.
Note that the first row of $A$ cannot be $\be_1$. 
Otherwise $\by_2, \by_3, \dots, \by_k$ will be always $\be_1$, contradicting $G=E_k$.
Thus, in order to obtain $G=E_k$, the first row of $A$ must be $\be_2$.
As the iterations proceed, with a similar argument, the first $(k-1)$ rows of $A$ must be $\be_2, \dots, \be_k$, respectively, in order to 
yield $G= E_k$. That is to say, 
$$
Y_k 
 = \begin{pmatrix}
 \bz \\ \bz A \\ \vdots \\ \bz A^{k-1}
 \end{pmatrix}
 = \begin{pmatrix}
 \bz \\ \ba_1 \\ \vdots \\ \ba_{k-1}
 \end{pmatrix}
 = \begin{pmatrix}
 \be_1 \\ \be_2 \\ \vdots \\ \be_{k}
 \end{pmatrix},
$$ where $\ba_i$ is the $i$-th row of the matrix $A$.
This relation fixes $\bz$ and  the first $(k-1)$ rows of $A$.
Consider $\by_{k+1} = \bz A^{k} = \be_k A = \ba_k$. 
Since $G=E_k$ has only $k$ rows,  $y_{k+1} = \ba_k$ must belong to $\Span{\Fq}{\be_1,\ldots,\be_k}.$ 
This means that there are $q^k$ choices of the $k$-th row $\ba_k$ of $A$.
There are no extra restrictions on $A$. That is to say, the vector $\bz $ together with matrix $A$ has the following structure:
\begin{equation}\label{eq:structure_u=1}
 \begin{pmatrix}
 \bz \\ A 
 \end{pmatrix}
 = 
 \begin{pmatrix}\begin{matrix}
I_k   
\\ 
\ba 
 \end{matrix}
 \,\,\vline \,\, \bzero
\\ \hline 
\hat{A}
 \end{pmatrix},
\end{equation}
where $\ba$ contains the first $k$ coordinates of $\ba_k$ which can be freely chosen from $\Fq^k,$ $\bzero$ is the zero matrix in $\Fq^{(k+1)\times(r-k)}$,  and $\hat{A}$ is an arbitrary matrix in $\Fq^{(r-k)\times r}$.
Therefore, there are in total $q^{(r-k)r + k}$ possible pairs $(\bz,A)$ that generate $G=E_k$ as output from Algorithm \ref{alg:SeqGen}. This implies that 
\begin{equation}\label{eq:|Ck1r|}
  |C_{k}^{(1,r)}| = A_q(r,k)|C_{E_k}^{(1,r)}|  = A_{q}(r,k) q^{r(r-k) + k}.  
\end{equation}
This is consistent with the statement in Theorem \ref{th:|C_k^ur|} for $u=1.$

\subsubsection{The cases $u\geq 2.$}
Similarly to the case of $u=1$, we proceed to count the cardinality of $C_{E_k}^{(u,r)}$ by describing the pattern which the pairs of matrices $(Z,A)$ in this set must follow.
This counting will involve the use of Ferrers diagrams, which we present in the Appendix\cn.
\begin{proposition}
	\label{prop:|Gammak|}
	The cardinality of the set
	$$
	C_{E_k}^{(u,r)} = \left\{
	(Z,A) \in \Fq^{u \times r} \times \Fq^{r \times r} 
	\mid Alg(Z,A, r) = E_k \right\},
	$$
	is given by
	\begin{equation}\label{eq:|Gammak|}
	|C_{E_k}^{(u,r)}| = 
	\begin{bmatrix}
	k + u - 1 \\ u-1
	\end{bmatrix}_q q^{r(r-k) + k}.
	\end{equation}
\end{proposition}
\begin{proof}
	Let $(Z,A)$ be a pair of matrices in $C_{E_k}^{(u,r)}.$
	This implies that $Z = (\hat{Z} | \bzero) \in \Fq^{u \times r}$ where $\hat{Z} \in \Fq^{u \times k}$ and $\bzero \in \Fq^{u \times (r-k)}$ is the zero matrix. 
    The matrix $A$ should satisfy the condition $\bv A \in \Span{\Fq}{E_k}$ for all $\bv \in \Span{\Fq}{E_k}.$
	This means that $A$ can be written as
	\begin{equation}\label{eq:StructA}
	A = 
	\begin{pmatrix}
	\begin{matrix}
	Y & \vline & \bzero
	\end{matrix}
	\\ 
	\hline
	\hat{A}
	\end{pmatrix},
	\end{equation}
	where $Y \in \Fq^{k \times k}$, $\bzero$ is the zero matrix in $\Fq^{k \times (r-k)}$ and $\hat{A} \in \Fq^{(r-k) \times r \cn}.$ Hence we have
	\begin{equation}\label{eq:StructM}
	\begin{pmatrix}
	Z \\ A
	\end{pmatrix}
	=  
	\begin{pmatrix}
	\begin{matrix}
	\begin{matrix} \hat{Z} \\ Y \end{matrix} &
	\vline & \bzero 
	\end{matrix} \\
	\hline
	\hat{A}
	\end{pmatrix} = \begin{pmatrix}
	W\\
	\hline
	\hat{A}
	\end{pmatrix},
	\end{equation}
	where $\bzero$ is the null matrix in $\Fq^{(u+k)\times(r-k)},$ $W=\begin{pmatrix}
	\begin{matrix} \hat{Z} \\ Y \end{matrix} &
	\vline & \bzero 
	\end{pmatrix} \in \Fq^{(u+k) \times r}.$ Notice that the submatrix $\hat{A} \in \Fq^{(r-k) \times r}$ has no impact on the execution of Algorithm \ref{alg:SeqGen}, so
	it can be chosen in $q^{r(r-k)}$ distinct ways.
	In the sequel we analyse the required properties of $W$ in detail by following the steps in Algorithm \ref{alg:SeqGen}.
	
	
	
    Recall from the discussion after Example \ref{Ex1} that $Z^{(0)}=Z$, while $G^{(0)}$ is obtained by sequentially extracting all the linearly independent row vectors in $Z^{(0)}.$
    For $j=1, \ldots, r-1$, $Z^{(j)} = G^{(j-1)}A,$ and $G^{(j)}$ is formed by all the rows in $Z^{(j)}$ which cannot be linearly expressed by rows in $G^{(0)}, \ldots, G^{(j-1)}$ and preceding rows in $Z^{(j)}.$
	It is clear that $Z^{(j)}$ is a sub-matrix of $Z^{(j-1)}A,$ and is thus a sub-matrix of $Z^{(0)}A^j = ZA^j.$     
	Suppose that in \eqref{eq: matrix_relation} the matrices $G^{(0)},\ldots, G^{(r-1)}$ have $k_0, \ldots, k_{r-1}$ rows, respectively. 
	Then $k_0 \geq k_1 \geq \ldots \geq k_{r-1}$ and $k_0 + k_1 +  \cdots +  k_{r-1} = k.$ Below we analyze the properties of $(Z,A)$ that gives as output the matrix $G=E_k$.
	
	As indicated in \eqref{eq: matrix_relation}, if $G = E_k$ then $G^{(0)} = E_{k_0}$ for some $k_0 \leq \min(u,k).$
    This implies that the matrix $Z^{(1)}=G^{(0)} A$ corresponds to the first $k_0$ rows of the matrix $A.$
	That is to say, $Z^{(1)}$ is formed by the first $k_0$ rows of the matrix $A$, and it corresponds to the rows ranging from $u+1$ to $u+k_0$ in matrix $W.$
	More generally, for $j=1, \dots, r-2$, the matrix $G^{(j)}$ is sequentially formed by the row vectors $\be_i$ with $k_0 + \cdots + k_{j-1} < i \leq k_0 + \ldots + k_{j}$.
    This implies that, the matrix $Z^{(j+1)} = G^{(j)} A,$ is formed by 
	the rows of $A$ in the same range, which corresponds to 
	the rows ranging from $u+(k_0 + \cdots + k_{j-1}) + 1$ to $u+(k_0 + \cdots + k_{j})$ in the matrix $W.$
	Hence we can summarize the relation between $W$, $Z^{(j)}$ and $G^{(j)}$ as follows:
	\begin{equation}\label{eq:S_struct}
	W = \begin{pmatrix}
	Z^{(0)}  \\ \hline 
	Z^{(1)}  \\ \hline 
	\vdots   \\ \hline 
	Z^{(r-1)}
	\end{pmatrix}
	\Longrightarrow \begin{pmatrix}
	G^{(0)}  \\ \hline 
	G^{(1)}  \\ \hline 
	\vdots   \\ \hline 
	G^{(r-1)}
	\end{pmatrix} = E_k,
	\end{equation} where matrices $G^{(j)}$ might be an empty matrix for $j$ larger than certain integer $t$ if $\be_k$ already occurs in $G^{(t)}$. 
	
	Denote by $\bw_i$ the $i$-th row of $W$ for $1\leq i\leq u+k$, and denote by $\bw_{i_1}, \dots, \bw_{i_k}$ the rows in $W$ corresponding to the first occurrence of $\be_1, \dots, \be_k$.
	From the above discussion we see that, for any $i$ with $i_j \leq i < i_{j+1}$, where $1\leq j<k$, the row $\bw_{i}$ can be linearly expressed by $\be_1, \dots, \be_{j}$, namely, $\bw_i $ can be any of the $q^{j}$ vectors of $\Span{\Fq}{\be_1, \dots, \be_j}$;
	and for those $i = i_j$ for $1\leq j\leq k$, the row $\bw_i$ must be the vector $\be_j$. 
	This indicates that there are $u$ ``free'' rows in $W$ in addition to the $k$ rows $\be_1, \dots, \be_k$.
	If $i_{k'} + l < i_{k'+1},$ the row $\bw_{i_{k'}+l}$ can be
	expressed as a linear combination of $\be_1, \dots, \be_{k'}$, indicating that $\bw_{i_{k'}+l}$ can be substituted by any vector such that its first $k'$ coordinates are randomly chosen from $\Fq$ and the remaining $(r-k')$ coordinates are set to zero.
    Since we know that the output must be $E_k,$ it must be that either $\be_k$ belongs to $Z$ or that $\be_k = \be_{i} A$ for some $i < k.$
	Then, similarly to the case of $u=1$, the last row of $W$, is always free to be any of the $q^k$ elements of $\calE_k.$
    
    Hence, we only need to consider the $(u-1)$ ``free'' rows in the matrix $W$.
	Below we provide a toy example to the above discussion before the remaining analysis.
	E.g., in the case of $k=4, u = 4$, the following matrix 
	$$
	W = 
	\begin{psmallmatrix}
	1   & 0   & 0   & 0   \\
	\oo & 0   & 0   & 0   \\
	0   & 1   & 0   & 0   \\
	\oo & \oo & 0   & 0  \\
	\\ \hline \\
	0   & 0   & 1   & 0   \\
	\oo & \oo & \oo & 0  \\
	\\ \hline \\
	0   & 0   & 0   & 1   \\
	\\ \hline \\
	\oo	& \oo & \oo & \oo 
	\end{psmallmatrix}
	=
	\begin{pmatrix}
	Z^{(0)} 
	\\ \hline 
	Z^{(1)} 
	\\ \hline 
	Z^{(2)} 
	\\ \hline 
	Z^{(3)} \\
	\end{pmatrix}
	\Longrightarrow
	E_4 = 
	\begin{pmatrix}
	G^{(0)} 
	\\ \hline 
	G^{(1)} 
	\\ \hline 
	G^{(2)} 
	\\ \hline
	G^{(3)} 
	\end{pmatrix}
	=
	\begin{psmallmatrix}
	1 & 0 & 0 & 0 \\
	0 & 1 & 0 & 0 
	\\ \hline \\
	0 & 0 & 1 & 0 
	\\ \hline \\
	0 & 0 & 0 & 1 \\
	\end{psmallmatrix},
	$$ where $G^{(3)}$ is an empty matrix, 
	and	each $\oo$ stands for a free variable in $\Fq.$
	We can consider the submatrix $W'$ made by the first $u+k-1$ rows of $W$.
	The matrix obtained from the free rows of this matrix resembles to a $(u-1)\times k$ Ferrers diagram $F_{W}$.
	Continuing with the example above we would obtain 
	$$
	F_{W'} =
	\begin{psmallmatrix}
	\oo & 0   & 0   & 0   \\
	\oo & \oo & 0   & 0   \\
	\oo & \oo & \oo & 0   
	\end{psmallmatrix}.
	$$
	Each of the possible diagrams are counted as $q^{|F_{W'}|}$ where $|\calF_{W'}|$ is the number of free variables in $\calF_{W'}.$ 
	Recall from (\ref{eq:ferrers_count_kxr}) that     
	the sum of all $n \times k$ Ferrers diagram counted by $q$ elevated to their weights is equal to $\begin{bsmallmatrix} n + k \\ k\end{bsmallmatrix}_q$. 
	Applying (\ref{eq:ferrers_count_kxr}) to our case, there are $\begin{bsmallmatrix} k + u - 1 \\ u - 1\end{bsmallmatrix}_q$ possible distinct matrices $W'.$
	
	To sum up, we have
	$$
	\begin{pmatrix}
	Z \\ A
	\end{pmatrix} = 
	\begin{pmatrix}
	\begin{matrix}
	\begin{matrix} W' \\ \bw_{u+k} \end{matrix} &
	\vline & \bzero 
	\end{matrix} \\
	\hline 
	\hat{A}
	\end{pmatrix},
	$$
	where the matrix $\hat{A}$ can be chosen in $q^{r(r-k)}$ distinct ways among the matrices in $\Fq^{(r-k) \times r}$, the vector $\bw_{k+u}$ can be chosen in $q^k$ ways among the vectors in $\calE_k$ and $W'$ can be chosen in $\gbinom{k + u - 1}{u - 1}_q$  distinct ways. This gives the cardinality 
	$$
	|C_{E_k}^{(u,k)}| = \begin{bmatrix} k + u - 1 \\ u - 1\end{bmatrix}_q q^{r(r-k) + k}.
	$$  
	The desired statement thus follows.
\end{proof}
From the results in Propositions \ref{cor:AqGammak} and \ref{prop:|Gammak|}, the desired statement in Theorem \ref{th:|C_k^ur|} directly follows. 

Before ending this section, 
below we calculate the probability that a uniformly distributed random choice of $(Z, A) \in \Fq^{u \times r} \times \Fq^{r \times r}$ expands the row space of $Z$ to the whole space $\Fq^r.$
\begin{corollary}\label{cor:Pr(C_k^u)}
    The probability that a random matrix $Y = \begin{psmallmatrix}
    Z \\ A
\end{psmallmatrix}$ uniformly drawn from $\Fq^{(u+r) \times r}$ produces a $k$-dimensional subspace $\Omega(Z,A)$ is given by 
\begin{equation}\label{eq:Pr(C_k^u)}
     \frac{|C_k^{(u,r)}|}{q^{r(r+u)}} =  \frac{H_q(r) H_q(k+u-1)}{H_q(r-k)H_q(k)H_q(u-1)}q^{-u(r-k)}.
\end{equation}
In particular, for $k=r$  the probability is given by
\begin{equation}\label{eq:Pr(C_r^u)}
    \frac{|C_r^{(u,r)}|}{q^{r(r+u)}} =  \frac{H_q(r+u-1)}{H_q(u-1)} \geq 1 - \frac{q^{-u+1}}{q-1}.
\end{equation}
\end{corollary}
\begin{proof}There are $q^{r(r+u)}$ matrices $Y$.
    It follows from Theorem \ref{th:|C_k^ur|} that 
    the probability is given by
    $$
        \frac{|C_k^{(u,r)}|}{q^{r(r+u)}} = A_q(r,k) \begin{bmatrix} k + u - 1 \\ u - 1\end{bmatrix}_q q^{r(r-k) + k - r(r+u)}.
    $$
    We can substitute $\begin{bmatrix} k + u - 1 \\ u - 1\end{bmatrix}_q = A_q(k + u - 1, u-1)/A_q(u-1,u-1)$
    and use (\ref{eq:Aq->Hq}) for each of the expression in $A_q.$
    In this way we obtain
    $$
         \frac{|C_k^{(u,r)}|}{q^{r(r+u)}} = \frac{H_q(r) H_q(k+u-1)}{H_q(r-k)H_q(k)H_q(u-1)}q^{-u(r-k)}.
    $$
    Using the convention $H_q(0) = 1$ and substituting $k=r$, we obtain the second expression in (\ref{eq:Pr(C_r^u)}).
    The approximation is obtained from 
    $$
        H_q(a)/H_q(a-b) \approx 1 - q^{-a} \frac{q^b - 1}{q-1} .
    $$
    In our case we have $a = r + u - 1$ and $b = r$, which concludes the proof.
\end{proof}

\subsection{Calculation of $C_r^{(u,r,t)}$}
In the previous subsection we calculate the probability that $\dim(\Omega(Z,A))=k.$
From its application in Section \ref{sec:Successful}, we are particularly interested in the case $k=r.$
To derive this probability we counted the pairs $(Z,A)$ such that $\Omega(Z,A)=\Omega_{r}(Z,A) = \Fq^r.$
From experimental results we observed that, in most of the cases, if $\dim(\Omega(Z,A)) = r,$ then $\dim(\Omega_{t}(Z,A)) = r$ for $t = \lceil \frac{r}{u} \rceil + 1.$
This means that, in most of the cases, Algorithm \ref{alg:SeqGen} outputs the same result on $(Z,A),t$ with $t$ much smaller than $r$.
In equation (\ref{eq:m>=2(d+t)r-3r}) of \cn Section \ref{sec:bdlrpc}, we have seen that $m \geq (2(d+t)-3)r$ where $t$ represents the number of expansions (i.e. the last input of Algorithm \ref{alg:SeqGen}).
In general $\lceil \frac{r}{u} \rceil + 1$ is much smaller than $r.$ It is clear that, using such a small value $t,$ allows for a smaller value of $m.$
In particular, asymptotically we can still correct an error with a probability arbitrarily close to $1$ and a value of $m$ such that $\lim_{n \to \infty}\frac{n}{m} = b$ for some $b > 0.$

Suppose that $u$ divides $r.$ For \cn $t = \frac{r}{u},$ it is already possible that $\Omega_t(Z,A) = \Omega(Z,A) = \Fq^r$.
Each of the matrices $Z,ZA, \ldots,ZA^{t-1}$ have $u$ rows. 
Combining all their rows we get $tu = r$ row vectors of $\Fq^r,$ if all these vectors are linearly  independent, then $\Omega_{r/u}(Z,A) = \Fq^r.$

We will show how, for large values of $u$ and $r,$ the probability that $\Omega_{t}(Z,A) = \Fq^r,$ conditioned on $\Omega(Z,A) = \Fq^r$ is lower bounded by $1-q^{-u/2}/(q-1)$ for $t = \lceil \frac{r}{u} \rceil + 1.$
The desired probability is obtained as the ratio between the cardinality of $C_r^{(u,r)} = \{(Z,A) \mid \Omega(Z,A) = \Fq^r \}$ and $C_r^{(u,r,t)} = \{(Z,A) \mid \Omega_t(Z,A) = \Fq^r\}$ noticing that $C_r^{(u,r,t)} \subseteq C_r^{(u,r)}.$
In Theorem \ref{th:|C_k^ur|} we already calculated $|C_r^{(u,r)}|.$
We will use a similar process to count $|C_r^{(u,r,t)}|.$
The proof will be structured in the following way.
\begin{itemize}
    \item By Proposition \ref{cor:Pr(C_k^u)}, we have $|C_k^{(u,r,t)}| = A_q(k,r)|C_{E_k}^{(u,r,t)}|$ for certain integer $t$.
    We are interested in the case that $r=k$, indicating $|C_r^{(u,r,t)}| = A_q(r,r)|C_{I_r}^{(u,r,t)}|$, where $E_k$ for $k=r$ becomes the identity matrix $I_r$ in $\Fq^{r\times r}$.
    \item As in the proof of Proposition \ref{prop:|Gammak|}, we consider the matrices $W = \begin{psmallmatrix}
        Z \\ A
    \end{psmallmatrix}$ that have output $E_r$ (i.e. the $r \times r$ identity matrix) through Algorithm \ref{alg:SeqGen}.
    We partition these matrices according to their associated Ferrers diagram.
    Then, we establish a relation between the number of inner cycles needed to obtain the output $E_r$ from a certain matrix W and its associated Ferrers diagram.
    \item Using this relation, we identify a particular subset of Ferrers diagrams such that, all the matrices matching those diagrams, need at most $t = \lceil \frac{r}{u} \rceil + 1$ inner cycles to obtain the desired output $E_r$.
    \item Although this subset \cn is not representing all the Ferrers diagrams with that property, when we count the set of all the matrices matching with a Ferrers  diagram in this subset\cn, this gives us already a good lower bound for ${|C_{E_r}^{(u,r,t)}|}$ as we will see in Theorem \ref{th:|Gammak^t|}.
\end{itemize}

We describe now more in detail the lower bound we give to $|C_{E_r}^{(u,r,t)}|.$
\begin{theorem}\label{th:|Gammak^t|}
   For an integer $t= \lceil \frac{r}{u} \rceil + 1$, one has  
    \begin{equation}\label{low:Gammak^t}
         |C_{E_r}^{(u,r,t)}| \geq \begin{bsmallmatrix}
            u + \frac{u}{2} \\ \frac{u}{2} 
        \end{bsmallmatrix}_q q^{ur-\frac{u^2}{2}}.
    \end{equation}
\end{theorem}
\begin{proof}
In Proposition \ref{prop:|Gammak|} we obtained $|C_{E_k}^{(u,r)}| = |C_{E_k}^{(u,r,r)}| = q^{r(r-k)+k} \begin{bsmallmatrix} k+u-1 \\ u-1 \end{bsmallmatrix}_q.$
We are interested in the case $k=r$. The formula of $|C_{E_k}^{(u,r)}|$ in Proposition  \ref{prop:|Gammak|}  was obtained by counting the possible matrices 
\begin{equation}
    W = 
    \begin{pmatrix}
        Z\\
        A
    \end{pmatrix}
    =
    \begin{pmatrix}
        Z^{(0)}\\
        Z^{(1)} \\
        \vdots \\
        Z^{(r-1)}
    \end{pmatrix},
\end{equation}
such that the each row is either $\be_i$ or a combination of the previous rows and the rows corresponding to $\be_i$ appear in the order from $\be_1$ to $\be_k.$

We managed to count all these matrices associating them with $(u-1)\times r$ Ferrers diagrams matching their structure.
From the Ferrers diagram associated to $W$, it is possible to count the number of steps $t$ needed such that 
$$
    E_r = \begin{psmallmatrix}
            G^{(0)}\\ \vdots \\ G^{(t-1)}
        \end{psmallmatrix}.
$$ 
As we have seen in the Appendix \cn, a $(u-1)\times r$ Ferrers diagram $F$ can be described by a sequence of $r$ integers $[f_1, \ldots, f_r]$ such that $u-1 \geq f_i \geq f_{i+1} \geq 0.$
Starting form this sequence it is possible to reconstruct a set of matrices $W \subseteq \Fq^{(u + k) \times r}$ having $F$ as their associated Ferrers diagram.
We start from $f_1$ and consider $f_0 = u-1, f_{r+1} = 0.$ 
If $f_1 < f_0$ then the first $f_1-f_0$ rows of $W$ will all be zero rows.
Otherwise, we consider the first index $i_1$ such that $f_{i_1} = f_0$ and $f_{i_1+1} < f_{i_1}.$
We add the rows $\be_1, \ldots, \be_{i_1}$ to $W$ followed by $f_{i_1} - f_{i_1 + 1}$ rows with $i_1$ initial $\oo$ followed by $r-i_1$ zeroes.
We proceed considering $f_j$ for $j \geq i_1 + 1.$
As long as $f_j = f_{i_1 + 1}$, we continue to add $\be_{i_1 +1}, \ldots, \be_{j}$ to $W$ until we find the next index $i_2$ such that $f_{i_2+1} < f_{i_2.}$
At this point we add $f_{i_2} - f_{i_2 + 1}$ identical rows with $i_2$ initial $\oo$ followed by $r- i_2$ zeroes.
We continue in this way until we reach $f_r.$
If $f_r > 0,$ then we add $f_r + 1$ rows entirely filled by $\oo,$ otherwise we add just one row of this type. 
(In Example \ref{ex:SSplit} we show this procedure on three different Ferrers diagrams.)

Recall that $Z^{(j)}$ is a submatrix of $Z^{(0)} A^j$ and is obtained from $G^{(j-1)} A.$
We can divide the matrix in blocks as follows.
The first block is always $Z = Z^{(0)}$ and corresponds to the first $u$ rows.
The second block is given by $Z^{(1)}= G^{(0)}A,$ where $G^{(0)}$ is the matrix obtained removing all the free rows of $Z^{(0)}$.
Suppose that $Z^{(0)}$ has $l_0$ free rows, then $G^{(0)}$ will have $u - l_0$ rows and $Z^{(1)} = G^{(0)}A$ will have the same number of rows which will correspond to the rows of $W$ having index between $u + 1$ and $2u - l_0.$

If  $l_0, \ldots, l_j$ are the number of free rows in $Z^{(0)}, \ldots, Z^{(j)}$, then the block $Z^{(j+1)} = G^{(j)}A$ will be constituted by the $u - (l_0 + \cdots + l_j)$ rows following the block $Z^{(j)}.$
In Example \ref{ex:SSplit}, we provide some examples of how to pass from a Ferrers diagram to its respective matrix form divided into the blocks $Z^{(j)}.$ 

We need to find a way to characterize the Ferrers diagrams such that 
\begin{equation}
    E_r = 
    \begin{psmallmatrix}
        G^{(0)} \\
        \vdots \\
        G^{(t-1)}
    \end{psmallmatrix}
\end{equation}
for $t \leq \lceil \frac{r}{u} \rceil + 1 .$
We will not completely characterize these diagrams but we will show a very large subset of Ferrers diagram with the desired property.
The matrices matching this subset, when counted with multiplicities, represents the vast majority of the possible cases.
This leads to a lower bound for $|C_{E_r}^{(r,u,t)}|$.

One simple way to count how many iterations are needed to generate the whole space $\Fq^r$ starting from a matrix $W$ having an associated $(u-1) \times r$ Ferrers diagram $F_W$ is the following.
First we add the row $(\oo \ldots \oo)$ on the bottom of $F_W.$
Then we can represent the Ferrers diagram as a $u \times r$ grid of squares and draw a line from the top-left corner to the bottom-right corner. 
This line will separate the cells containing a $\oo$ from the cells containing a zero.
The line is always constituted of precisely $r$ horizontal segments and $u$ vertical segments.
Beginning from the top-left corner, we move along the line by $u$ segments and mark the point we reach.
In subsequent steps, we start from the last marked point, we continue to move along the line by a number of segments corresponding to the number of horizontal segments we covered in the previous step and we mark the new point that we reach.

This procedure will divide the line into at most $r$ chunks.
We repeat this procedure until all the horizontal segments are covered.
The number of chunks corresponds to the number of blocks $Z^{(j)} $ in $W$ \cn needed to generate $\Fq^r.$
In Example \ref{ex:Line} we consider some Ferrers diagram and show their respective line partition according to the procedure described above.

We can identify a subset of Ferrers diagrams covered by a small number of steps.

For $t = \lceil \frac{r}{u} \rceil + 1$ we have $t-1 \geq \frac{r}{u}$ with equality only when $u$ divides $r$.
Assuming that $u$ divides $r$ represents the most challenging scenario and can then be used to set a lower bound.

Let $F$ be the $u \times r$ Ferrers diagram described by the sequence of non-decreasing positive integers, i.e., $F = [f_1, \ldots, f_r],$ where $f_i = u$ if $i \leq r-u$ and $f_i= \frac{u}{2}$ otherwise.
Assume $\widehat{F} = \{ F' \in \calF_{u,r} \mid F' \geq F\}$ is the set of all the possible Ferrers diagrams bigger than $F$ as in Lemma \ref{lm:FerrersSpecialCount}.
If we depict the $u \times r$ Ferrers diagrams as paths over a $u \times r$ grid, $\widehat{F}$ is the set of all paths that are always on top or equal to $F.$
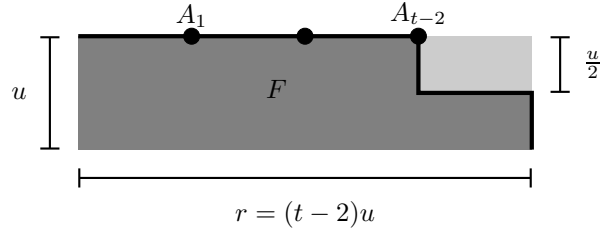
\begin{center}
  \begin{tikzpicture}[scale=0.75]
  \filldraw[fill=black!50, draw=none] (0,0) rectangle (8,2);

  \filldraw[fill=black!20, draw=none] (6,1) rectangle (8,2);

  \draw[ultra thick, black] (0,2) |- (6,2) |- (6,1) |- (8,1) |- (8,0);
  
  \filldraw[fill=black] (2,2) circle (4pt);
  \filldraw[fill=black] (4,2) circle (4pt);
  \filldraw[fill=black] (6,2) circle (4pt);

  \draw[black] (2, 2) node [anchor = south] {$A_1$};
  \draw[black] (6, 2) node [anchor = south] {$A_{t-2}$};
  \draw[black] (3.5, 1.4) node [anchor = north] {$F$};

    \draw[thick, |-|] (-0.5,0) -- (-0.5,2) node[midway, anchor=east, xshift=-5pt] {$u$};

    \draw[thick, |-|] (0,-0.5) -- (8,-0.5) node[midway, anchor = north, yshift= -5pt] {$r = (t-2)u$};

    \draw[thick, |-|] (8.5,1) -- (8.5,2) node[midway, anchor = west, xshift= 5pt] {$\frac{u}{2}$};
\end{tikzpicture}
\end{center}
In dark gray, we represented the Ferrers diagram $F.$
All the Ferrers diagrams in $\widehat{F}$ will cover that area.
In light gray, we represented the part that can change between two different elements of $\widehat{F}$.

Assume $F' \geq F.$ Then the corresponding path will go through $A_1, \ldots, A_{t-2}$ in $t-2$ steps.
The first part of the path consists of $t-2$ horizontal steps each of length $u$ covering a total of $r-u$ horizontal segments.
In the following steps, we can move by at most $\frac{u}{2}$ vertical segments, then we will always move by at least $\frac{u}{2}$ horizontal segments until we reach the right border.
It takes at most two more steps to reach the right border from $A_{t-2}.$ In total we will then need at most $t$ steps.

Let $T$ be the set of all the matrices $W$ that match with a Ferrers diagram in $\widehat{F},$ we will have 
\begin{equation}
    |C_{E_r}^{(u,r,t)}| \geq |T| = \sum_{F' \in \widehat{F}} q^{|F'|}.
\end{equation}
Finally, we can apply Lemma \ref{lm:FerrersSpecialCount} to $\widehat{F}$ to obtain
\begin{equation}
    |C_{E_r}^{(u,r,t)}| \geq
    \begin{bsmallmatrix} 
        u + \frac{u}{2} \\    
        \frac{u}{2}
    \end{bsmallmatrix}_q
    q^{ur - \frac{u^2}{2}}. 
\end{equation}
\end{proof}
To help the reader, we provide some examples of Ferrers diagrams, their relative matrix conversion, and division into blocks $Z^{(0)}, \ldots, Z^{(t-1)}.$
\begin{example}\label{ex:SSplit}
    Let $u = 4$ and $r = 8$. 
    The Ferrers diagram below can be converted into a matrix $W$ with the following structure
    \begin{equation}\label{eq:SFromFerrerExample}
        \begin{psmallmatrix}
        \oo & \oo & \oo & \oo & 0   & 0   & 0   & 0   \\
        \oo & \oo & \oo & \oo & 0   & 0   & 0   & 0   \\
        \oo & \oo & \oo & \oo & \oo & \oo & \oo & \oo
        \end{psmallmatrix}
        \quad
        \rightarrow
        \quad
        W =
        \begin{psmallmatrix}
		1   & 0   & 0   & 0   & 0   & 0   & 0   & 0   \\
		0   & 1   & 0   & 0   & 0   & 0   & 0   & 0   \\
		0   & 0   & 1   & 0   & 0   & 0   & 0   & 0   \\
		0   & 0   & 0   & 1   & 0   & 0   & 0   & 0   \\
		  \\ \hline \\
        \oo & \oo & \oo & \oo & 0   & 0   & 0   & 0   \\
        \oo & \oo & \oo & \oo & 0   & 0   & 0   & 0   \\
		0   & 0   & 0   & 0   & 1   & 0   & 0   & 0   \\
		0   & 0   & 0   & 0   & 0   & 1   & 0   & 0   \\
		\\ \hline \\
		0   & 0   & 0   & 0   & 0   & 0   & 1   & 0   \\
		0   & 0   & 0   & 0   & 0   & 0   & 0   & 1   \\
		\\ \hline \\
        \oo & \oo & \oo & \oo & \oo & \oo & \oo & \oo \\
		\oo & \oo & \oo & \oo & \oo & \oo & \oo & \oo
        \end{psmallmatrix}.
    \end{equation}
    We divided the resulting matrix in blocks.
    These blocks corresponding to $Z^{(0)}, \ldots, Z^{(t-1)}$ while, if we remove the free rows, the blocks correspond to the matrices $G^{(0)}, \ldots, G^{(t-1)}.$

    An extreme example is provided by the following Ferrers diagram
    \begin{equation}\label{eq:SFromFerrerExample2}
        \begin{psmallmatrix}
        \oo & \oo & \oo & \oo & \oo & \oo & \oo & \oo  \\
        \oo & \oo & \oo & \oo & \oo & \oo & \oo & \oo  \\
        \oo & \oo & \oo & \oo & \oo & \oo & \oo & \oo
        \end{psmallmatrix}
        \quad
        \rightarrow
        \quad
        W =
        \begin{psmallmatrix}
		1   & 0   & 0   & 0   & 0   & 0   & 0   & 0   \\
		0   & 1   & 0   & 0   & 0   & 0   & 0   & 0   \\
		0   & 0   & 1   & 0   & 0   & 0   & 0   & 0   \\
		0   & 0   & 0   & 1   & 0   & 0   & 0   & 0   \\
		\\ \hline \\
		0   & 0   & 0   & 0   & 1   & 0   & 0   & 0   \\
		0   & 0   & 0   & 0   & 0   & 1   & 0   & 0   \\		
		0   & 0   & 0   & 0   & 0   & 0   & 1   & 0   \\
		0   & 0   & 0   & 0   & 0   & 0   & 0   & 1   \\
		\\ \hline \\
        \oo & \oo & \oo & \oo & \oo & \oo & \oo & \oo \\
		\oo & \oo & \oo & \oo & \oo & \oo & \oo & \oo \\
		\oo & \oo & \oo & \oo & \oo & \oo & \oo & \oo \\
		\oo & \oo & \oo & \oo & \oo & \oo & \oo & \oo
        \end{psmallmatrix}.
    \end{equation}
    In this case  $G^{(0)}$ and $G^{(1)}$  coincide with $Z^{(0)}$ and $Z^{(1)},$ most important, the rows in the first two blocks already generates $\Fq^r$, while in (\ref{eq:SFromFerrerExample}) we needed to consider the first $3$ blocks to span $\Fq^r.$

    An opposite example can be given by the following
    \begin{equation}\label{eq:SFromFerrerExample3}
        \begin{psmallmatrix}
        \oo & 0   & 0   & 0   & 0   & 0   & 0   & 0   \\
        \oo & 0   & 0   & 0   & 0   & 0   & 0   & 0   \\
        \oo & \oo & \oo & \oo & 0   & 0   & 0   & 0
        \end{psmallmatrix}
        \quad
        \rightarrow
        \quad
        W =
        \begin{psmallmatrix}
		  1   & 0   & 0   & 0   & 0   & 0   & 0   & 0   \\
		\oo & 0   & 0   & 0   & 0   & 0   & 0   & 0   \\
		\oo & 0   & 0   & 0   & 0   & 0   & 0   & 0   \\
		0   & 1   & 0   & 0   & 0   & 0   & 0   & 0   \\
		\\ \hline \\
		0   & 0   & 1   & 0   & 0   & 0   & 0   & 0   \\
		0   & 0   & 0   & 1   & 0   & 0   & 0   & 0   \\
		\\ \hline \\
		\oo & \oo & \oo & \oo & 0   & 0   & 0   & 0	  \\
		0   & 0   & 0   & 0   & 1   & 0   & 0   & 0   \\
		\\ \hline \\
		0   & 0   & 0   & 0   & 0   & 1   & 0   & 0   \\
		\\ \hline \\		
		0   & 0   & 0   & 0   & 0   & 0   & 1   & 0   \\
		\\ \hline \\
		0   & 0   & 0   & 0   & 0   & 0   & 0   & 1   \\
		\\ \hline \\
        \oo & \oo & \oo & \oo & \oo & \oo & \oo & \oo 
        \end{psmallmatrix}.
    \end{equation}
    In this case we need to consider $Z^{(0)}, \ldots, Z^{(5)}$ to span $\Fq^r.$
    This means that 
    $$\Fq^r = \sum_{i= 0}^5\RowSpan(ZA^i),$$ which is close to the theoretical upper bound of $r-1 = 7.$
\end{example}

The first thing to notice from the Ferrers diagrams in Example \ref{ex:SSplit} is that, if $(Z,A)$ form a matrix that match a Ferrers diagram with an high weight (i.e. a big number of $\oo$), then after few iterations $Z, ZA, \ldots, ZA^j$ will span the whole space $\Fq^r.$
The second thing to notice is that, when counted with multiplicities, these matrices are disproportionately more than the matrices with a low weight Ferrers diagram.

Starting form the first row, each free row represents an interruption to the expansion.
If an interruption happens early, all the following steps will be shorter by one, so it will take more steps to reach $\Fq^r.$
For example, if in $Z$ there are only $u-1$ out of $u$ rows which are linearly independent, we will need at least $\frac{r}{u-1}$ steps to cover $\Fq^r$ instead of $\frac{r}{u}$ as we could expect from a full rank matrix.

\begin{example}\label{ex:Line}
  We illustrate the line procedure we used in the proof of Theorem \ref{th:|Gammak^t|} on the Ferrers diagrams we introduced in the previous example.
\begin{table}[H]
  \centering
  \begin{tabular}{cc}
		\adjustbox{valign=c}{%
			\begin{tikzpicture}[scale=0.5]
				\foreach \i [count=\j] in {4,4,8,8}{
					\foreach \k in {1,...,\i}{
					\filldraw[fill=black!20] (\k,-\j) rectangle (\k+1,-\j+1);
					}
				}
				
				\draw[ultra thick, black] (1,0) |- (5,0) |- (7,-2) |- (9,-2) |- (9,-4);
				
				\filldraw[fill=black] (5,0) circle (5pt);
				\filldraw[fill=black] (7,-2) circle (5pt);
				\filldraw[fill=black] (9,-2) circle (5pt);
			\end{tikzpicture}%
        }
    &
    $ \to 
    \begin{psmallmatrix}
		1   & 0   & 0   & 0   & 0   & 0   & 0   & 0   \\
		0   & 1   & 0   & 0   & 0   & 0   & 0   & 0   \\
		0   & 0   & 1   & 0   & 0   & 0   & 0   & 0   \\
		0   & 0   & 0   & 1   & 0   & 0   & 0   & 0   \\
		  \\ \hline \\
        \oo & \oo & \oo & \oo & 0   & 0   & 0   & 0   \\
        \oo & \oo & \oo & \oo & 0   & 0   & 0   & 0   \\
		0   & 0   & 0   & 0   & 1   & 0   & 0   & 0   \\
		0   & 0   & 0   & 0   & 0   & 1   & 0   & 0   \\
		\\ \hline \\
		0   & 0   & 0   & 0   & 0   & 0   & 1   & 0   \\
		0   & 0   & 0   & 0   & 0   & 0   & 0   & 1   \\
		\\ \hline \\
        \oo & \oo & \oo & \oo & \oo & \oo & \oo & \oo \\
		\oo & \oo & \oo & \oo & \oo & \oo & \oo & \oo
\end{psmallmatrix}.$    
  \end{tabular}
  \caption*{An example illustrating (\ref{eq:SFromFerrerExample})}
\end{table}

\begin{table}[H]
  \centering
  \begin{tabular}{cc}
    \adjustbox{valign=c}{%
			\begin{tikzpicture}[scale=0.5]
				\foreach \i [count=\j] in {8,8,8,8}{
					\foreach \k in {1,...,\i}{
					\filldraw[fill=black!20] (\k,-\j) rectangle (\k+1,-\j+1);
					}
				}
				
				\draw[ultra thick, black] (1,0) |- (5,0) |- (9,0) |- (9,-4);
				
				\filldraw[fill=black] (5,0) circle (5pt);
				\filldraw[fill=black] (9,0) circle (5pt);
			\end{tikzpicture}%
        }
    &
    $ \to 
    \begin{psmallmatrix}
		1   & 0   & 0   & 0   & 0   & 0   & 0   & 0   \\
		0   & 1   & 0   & 0   & 0   & 0   & 0   & 0   \\
		0   & 0   & 1   & 0   & 0   & 0   & 0   & 0   \\
		0   & 0   & 0   & 1   & 0   & 0   & 0   & 0   \\
		\\ \hline \\
		0   & 0   & 0   & 0   & 1   & 0   & 0   & 0   \\
		0   & 0   & 0   & 0   & 0   & 1   & 0   & 0   \\		
		0   & 0   & 0   & 0   & 0   & 0   & 1   & 0   \\
		0   & 0   & 0   & 0   & 0   & 0   & 0   & 1   \\
		\\ \hline \\
        \oo & \oo & \oo & \oo & \oo & \oo & \oo & \oo \\
		\oo & \oo & \oo & \oo & \oo & \oo & \oo & \oo \\
		\oo & \oo & \oo & \oo & \oo & \oo & \oo & \oo \\
		\oo & \oo & \oo & \oo & \oo & \oo & \oo & \oo
\end{psmallmatrix}.$    
  \end{tabular}
  \caption*{An example illustrating (\ref{eq:SFromFerrerExample2})}
\end{table}

\begin{table}[H]
  \centering
  \begin{tabular}{cc}
    \adjustbox{valign=c}{%
      \begin{tikzpicture}[scale=0.5]
        \foreach \i [count=\j] in {1,1,4,8}{
          \foreach \k in {1,...,\i}{
            \filldraw[fill=black!20] (\k,-\j) rectangle (\k+1,-\j+1);
          }
        }
        \draw[ultra thick, black] (1,0) |- (2,0) |- (5,-2) |- (9,-3) |- (9,-4);
        \foreach \x/\y in {3/-2, 5/-2, 6/-3, 7/-3, 8/-3, 9/-3}{
          \filldraw[fill=black] (\x,\y) circle (5pt);
        }
      \end{tikzpicture}%
    }
    &
    $ \to 
    \begin{psmallmatrix}
		1   & 0   & 0   & 0   & 0   & 0   & 0   & 0   \\
		\oo & 0   & 0   & 0   & 0   & 0   & 0   & 0   \\
		\oo & 0   & 0   & 0   & 0   & 0   & 0   & 0   \\
		0   & 1   & 0   & 0   & 0   & 0   & 0   & 0   \\
		\\ \hline \\
		0   & 0   & 1   & 0   & 0   & 0   & 0   & 0   \\
		0   & 0   & 0   & 1   & 0   & 0   & 0   & 0   \\
		\\ \hline \\
		\oo & \oo & \oo & \oo & 0   & 0   & 0   & 0	  \\
		0   & 0   & 0   & 0   & 1   & 0   & 0   & 0   \\
		\\ \hline \\
		0   & 0   & 0   & 0   & 0   & 1   & 0   & 0   \\
		\\ \hline \\		
		0   & 0   & 0   & 0   & 0   & 0   & 1   & 0   \\
		\\ \hline \\
		0   & 0   & 0   & 0   & 0   & 0   & 0   & 1   \\
		\\ \hline \\
        \oo & \oo & \oo & \oo & \oo & \oo & \oo & \oo 
    \end{psmallmatrix}$    
  \end{tabular}
  \caption*{An example illustrating (\ref{eq:SFromFerrerExample3})}
\end{table}
\end{example}
\FloatBarrier
We provide another example to clarify the last part of the proof of Theorem  \ref{th:|Gammak^t|}.
 \begin{example} \label{ex:48}
     Let  $u=4$ and $r=8$, consider the following set of examples 
\begin{equation}\label{eq:GoodSetFerrers}
    \begin{psmallmatrix}
        \oo & \oo & \oo & \oo & \og & \og & \og & \og \\
        \oo & \oo & \oo & \oo & \og & \og & \og & \og \\
        \oo & \oo & \oo & \oo & \oo & \oo & \oo & \oo
    \end{psmallmatrix},
\end{equation}
where $\og$ can either be a $0$ or $\oo.$

Consider $S$ a matrix with an associated Ferrers diagram matching this pattern.
We claim that $\Fq^r = \sum_{i=0}^2 \RowSpan(Z^{(i)}) = \Omega_3(Z,A).$

To see that we can consider the associated grid
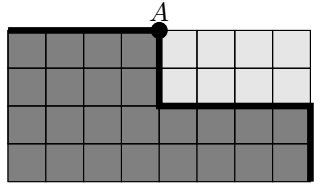
\begin{center}
    \begin{tikzpicture}[scale=0.5]
  \foreach \i in {1,2,...,8}{
    \foreach \j in {-4,-3,-2,-1}{
      \filldraw[fill=black!10] (\i,\j) rectangle (\i+1,\j+1);
    }
  }

  \foreach \i [count=\j] in {4,4,8,8}{
   \foreach \k in {1,...,\i}{
      \filldraw[fill=black!50, thin] (\k,-\j) rectangle (\k+1,-\j+1);
    }
  }

  \draw[step=1,black] (1,-3) grid (9,0);

  \draw[line width=2.5pt, black] (1,0) |- (5,0) |- (7,-2) |- (9,-2) |- (9,-4);

  \filldraw[fill=black] (5,0) circle (6pt);
  \draw[black] (5, 0) node [anchor = south] {$A$};
  
\end{tikzpicture}
.
\end{center}
Any path that stays above the path we draw will be such that the first marked point is $A.$
In the subsequent steps we can move by at most two vertical segments, then we will always move right \cn by at least two horizontal segments until we reach the right edge.
Since we already covered $4$ horizontal segments in the first step, with two more steps of length $\min\{2, 8-c\}$ where $c$ is the number of horizontal steps we did so far, we will reach the the right edge in at most $3$ steps.
Notice that $3 = \frac{r}{u} + 1.$

Let us call $T$ the set of all the matrices $W$ \cn having a Ferrers diagram that match one of the possible Ferrers diagram in (\ref{eq:GoodSetFerrers}).
The possible different paths we can consider correspond to all the possible Ferrers diagram we can draw inside the top-left rectangle. 
This can be counted as all the possible $2 \times 4$ Ferrers diagrams.
We can count how many matrices are in the set $T$ applying Lemma \ref{lm:FerrersSpecialCount} with parameters $s=2, u=t=4, k=8$  obtaining \cn
\begin{equation}\label{eq:|T|example}
    |T| = q^{24} \begin{bsmallmatrix}
        6 \\ 2
    \end{bsmallmatrix}_q.
\end{equation}
If we compare it to the number of all the possible matrices $W$ that gives us output $E_r$ we obtain
$$
    \frac{|T|}{|C_{E_8}^{(4,8)}|} 
    =  \frac{\begin{bsmallmatrix}
        6 \\ 2
    \end{bsmallmatrix}_q}{\begin{bsmallmatrix}
        11 \\ 3
    \end{bsmallmatrix}_q}
    q^{16}
    =
    \frac{H_q(6)H_q(8)H_q(3)}{H_q(4)H_q(2)H_q(11)}.
$$
Which is lower bounded by $\frac{H(6)}{H(2)} \approx 1 - \frac{q^{-2}}{q-1}.$
 \end{example}
This example shows us how to prove the probability in the following corollary.
\begin{corollary}\label{prop:Prob_t}
    Let $Z \in \Fq^{u \times r}$ and $A \in \Fq^{r\times r}$ be two matrices uniformly independently chosen in their respective domains.
    For $t \geq \lceil \frac{r}{u} \rceil + 1,$ the probability
    $$
        {\rm Prob}\left(\Omega_{t}(Z,A) = \Fq^r \mid \Omega(Z,A) = \Fq^r\right ) \geq 1 - \frac{q^{-u/2}}{q-1}. 
    $$
\end{corollary}
\begin{proof}
    Observe that $C_r^{(u,r,t)} \subseteq C_r^{(u,r)}$ we have 
    \begin{equation}\label{eq:P(Omega^t)}
        P(\Omega_{t}(Z,A) = \Fq^r \mid \Omega(Z,A) = \Fq^r) = \frac{|C_r^{(u,r,t)}|}{|C_r^{(u,r)}|} = \frac{|C_{E_r}^{(u,r,t)}|}{|C_{E_r}^{(u,r)}|}. 
    \end{equation}   
    Theorem  \ref{th:|Gammak^t|} provides us with the lower bound
    \begin{equation}
    |C_{E_r}^{(u,r,t)}| \geq
    \begin{bsmallmatrix} 
        u + \frac{u}{2} \\    
        \frac{u}{2}
    \end{bsmallmatrix}_q
    q^{ur - \frac{u^2}{2}}. 
    \end{equation}
    Substituting it in (\ref{eq:P(Omega^t)}) we get
    \begin{align}
     \frac{|C_{E_r}^{(u,r,t)}|}{|C_{E_r}^{(u,r)}|} 
    & \geq 
    \frac{\begin{bsmallmatrix}
        u + \frac{u}{2} \\    
        \frac{u}{2}
    \end{bsmallmatrix}_q}
    {\begin{bsmallmatrix}
        r+u-1 \\ u-1
    \end{bsmallmatrix}_q}
    q^{(u-1)r - \frac{u^2}{2}} \notag \\
    & = 
    \frac
    {H_q(u + \frac{u}{2})H_q(u-1)H_q(r)}
    {H_q(\frac{u}{2})H_q(u)H_q(r+u-1)},
    \end{align}
    which can be further lower bounded as
    \begin{equation}
    \frac{|C_{E_r}^{(u,r,t)}|}{|C_{E_r}^{(u,r)}|}
    \geq
    \frac{H_q(u + \frac{u}{2})}{H_q(\frac{u}{2})} = \prod_{i = \frac{u}{2} + 1}^u (1- q^{-i}) \approx 1 - \frac{q^{-\frac{u}{2}}}{q-1}.
\end{equation}
\end{proof}

\section{Conclusion} \label{sec:conclusion}
This paper presents a novel family of bounded-degree low-rank parity-check (BD-LRPC) codes, which exhibits remarkable error correction capabilities. The proposed BD-LRPC codes are constructed by leveraging the inherent properties of the subspace $\calV_{\alpha,d}=\Span{\Fq}{1,\alpha,\dots,\alpha^{d-1}}$, from which the decoding radius can approach the Singleton bound.

Our contributions in this paper are twofold. First, we introduce this novel family of BD-LRPC codes, demonstrating their superior error correction capability compared to LRPC codes and their ability to approach the Singleton bound.
Second, we formulate an intriguing counting problem (Problem \ref{Pb:C_k}) that arises from the study of these codes.
We provide a sophisticated approach to tackling this problem using Ferrers diagrams for the cases that $t=r$ and $t=\lceil \frac{r}{u} \rceil+1$. The counting problem seems to have its own interest in combinatorics and may have applications in other domains.
We cordially invite interested readers to find a tighter bound or an exact solution for $\lceil \frac{r}{u} \rceil \leq t < r.$

\section{ Appendix: Ferrers Diagrams}

 Let $\bg_1, \ldots, \bg_k$ be a base of $\Omega \subseteq \Fq^r$. 
 Consider $G \in \Fq^{k \times r}$ to be the matrix having for rows $\bg_1, \ldots, \bg_k$.
 If $A \in GL_q(k)$ is an invertible matrix, the rows of $AG$ are generating the same subspace $\Omega.$
An alternative way to count the subspaces of dimension $k$ in $\Fq^n,$ is to count the number of full rank matrices in $\Fq^{n \times k}$ in column reduced echelon form \cite{knuth71}\cn.
Recall that a matrix $A \in \Fq^{n \times k}$ of rank $k$ is in column reduced echelon form if:
\begin{itemize}
    \item The topmost nonzero entry of each column is a $1$ (called a leading one).
    \item The other entries of a row where there is a leading one are all zeroes.
    \item If we denote by $n_i$ the position of the leading one of the $i$-th column, then $n_1 < n_2 < \cdots < n_k \leq n.$
\end{itemize}
As an example the following matrix is in column reduced echelon form
$$
    A = 
    \begin{psmallmatrix}
        0    & 0   & 0   \\
		1    & 0   & 0   \\ 
        0    & 1   & 0   \\
        \oo  & \oo & 0   \\
        \oo  & \oo & 0   \\
        0    & 0   & 1   \\
        \oo  & \oo & \oo 		
    \end{psmallmatrix},
$$ where each $\bullet$ represents a free value in $\Fq$. 
The number of matrices following the same pattern of $A \in \Fq^{7 \times 3}$ as above is $q^7$ where $7$ is the number of $\bullet$ in $A$.

Let us now consider the alternative counting. Let $\calV \subseteq \Fq^n$ be a subspace of dimension $k$ and let $G \in \Fq^{n \times k}$ be a matrix such that its columns form a basis of $\calV.$
The  column reduced echelon form of $G$ is a basis of $\calV$ and it is unique.
All vector spaces of dimension $k$  correspond to a unique generator in 
column reduced echelon form.

Removing all the rows containing a leading one from a reduced column echelon form matrix in $\Fq^{n \times k}$ gives a Ferrers diagram.
If we consider only matrices of full rank $k,$ there will be exactly $k$ leading ones, then the resulting Ferrers diagram will resemble to an $(n-k) \times k$ matrix.
In our example, the Ferrers diagram associated to $A$ is the $4 \times 3$ diagram
\begin{equation}\label{eq:FerrersExample}
    F=
    \begin{psmallmatrix}
        0    & 0   & 0   \\
        \oo  & \oo & 0   \\
        \oo  & \oo & 0   \\
        \oo  & \oo & \oo 		
    \end{psmallmatrix}.
\end{equation}   
Each $n \times k$ Ferrers diagram $F$ has a natural representation as a non-increasing sequence $[f_1, \ldots, f_k]$ of length $k$ of integers $0 \leq f_i \leq n.$
The sequence representation of (\ref{eq:FerrersExample}) is $[3,3,1].$

Consider an $n \times k$ Ferrers diagram $F$ as a matrix.
Denote by $a_{i,j}$ the element at the intersection of the $i$-th row with the $j$-th column.
Whenever $i \neq n$ and $j \neq 1,$ if $a_{i,j} = \bullet,$ then $a_{i+1, j} = \bullet$ and $a_{i,j-1} = \bullet.$ 
If $F$ is a $n \times k$ Ferrers diagram, its weight $|F|$ correspond to the number of free variables $\bullet.$

It is easy to see that, for each $(n-k) \times k$ Ferrers diagram $F,$ there are exactly $q^{|F|}$ possible column reduced echelon matrices $A \in \Fq^{k \times r}$ that have $F$ as the associated Ferrers diagram.
Let $\calF_{n-k,k}$ be the set of all $(n-k) \times k$ Ferrers diagrams. Knuth \cite{knuth71} showed the following relation between Ferrers diagram and the Gaussian binomial coefficients:
\begin{equation}\label{eq:ferrers_count_kxr-k}
    \sum_{F \in \calF_{n-k,k}} q^{|F|} = \gbinom{n}{k}_q = A_q(n,k)/A_q(k,k)
\end{equation} since each column reduced echelon matrix corresponds to a $k$-dimensional subspace in $\Fq^n$.
We will see that this formula is important to attack Problem \ref{Pb:C_k} in Section \ref{sec:ExpRowSub}.

Considering the set $\calF_{n,k}$ of $n \times k$ Ferrers diagrams as non-decreasing sequences, it is natural to define a partial order relation over this set.
Let $F = [f_1, \ldots, f_k] \in \calF_{n,k}$ and $F' = [f'_1, \ldots, f'_k] \in \calF_{n,k}$ be two Ferrers diagrams, we define the partial order relation $F \leq F'$ if $f_i \leq f'_i$ for $1\leq i\leq k$.
We will denote by $\widehat{F} = \{ F' \in  \calF_{n,k} \mid F \leq F'\},$ the set of all the Ferrers diagrams bigger than $F.$  We can combine \eqref{eq:ferrers_count_kxr-k} and the partial order on Ferrers diagrams to determine the number of a certain class of matrices.

\begin{lemma}\label{lm:FerrersSpecialCount}
    Let $F = [f_1,\ldots,f_k]$ be a $u \times k$ Ferrers diagram such that $f_i = u$ for $i \leq k-t$ and $f_i = u-s$ otherwise.
    Then we have 
    \begin{equation}\label{eq:ferrers_count_kxr}
        \sum_{F' \in \widehat{F}} q^{|F'|} = 
        \gbinom{s+t}{s}_q q^{uk - st}
    \end{equation}
    where $\widehat{F} = \{F'  \in \calF_{u,k} \mid F \leq F'\}.$
\end{lemma}
\begin{proof}
    For $F' \in \hat{F},$ denote $F' = F + L$ where the addition between $F$ and $L$ is the component wise sum.
    
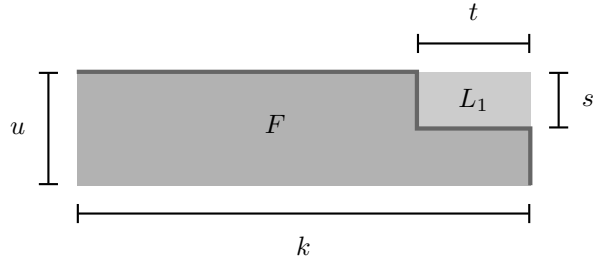
\begin{center}
    	\begin{tikzpicture}[scale=0.75]
    	\filldraw[fill=black!30, draw=none] (0,0) rectangle (8,2);
    	
    	\filldraw[fill=black!20, draw=none] (6,1) rectangle (8,2);
    	
    	\draw[ultra thick, black!60] (0,2) |- (6,2) |- (6,1) |- (8,1) |- (8,0);
    	
    	\draw[black] (3.5, 1.4) node [anchor = north] {$F$};
    	\draw[black] (7, 1.5) node [anchor = center] {$L_1$};
    	
    	\draw[thick, |-|] (-0.5,0) -- (-0.5,2) node[midway, anchor=east, xshift=-5pt] {$u$};
    	\draw[thick, |-|] (0,-0.5) -- (8,-0.5) node[midway, anchor = north, yshift= -5pt] {$k$};
    	\draw[thick, |-|] (8.5,1) -- (8.5,2) node[midway, anchor = west, xshift= 5pt] {$s$};
    	\draw[thick, |-|] (6,2.50) -- (8,2.5) node[midway, anchor = south, yshift= +5pt] {$t$};
    	\end{tikzpicture}
    \end{center}
    We write $L$ as the concatenation of two sequences $L_0|| L_1$, where $L_0$ has length $k-t$ and $L_1$ has length $t$.
    It is clear that $L_0 = [0, \ldots, 0]$ and $L_1$ is a non increasing sequence of integers upper bounded by $s$.
     Hence $L_1$ describes an $s \times t$ Ferrers diagram. From (\ref{eq:ferrers_count_kxr-k}) we see that
   $$
       \sum_{L_1 \in \calF_{s,t}} q^{|L_1|} = \gbinom{s+t}{t}_q.
   $$    
    Combining the fact $q^{|F|} = q^{uk -st}$, we obtain the desired result.
    \end{proof}

\bibliographystyle{abbrv}
\bibliography{bibliography.bib}

\end{document}